\documentclass[10pt, a4paper, oneside]{article}
\usepackage[italian,english]{babel}
\usepackage[a4paper,top=30mm,bottom=40mm,inner=30mm,outer=20mm]{geometry}
\usepackage{listings}
\usepackage{color}
\usepackage{setspace}
\usepackage{hyperref}
\usepackage{acro}
\usepackage{amsmath}
\usepackage{amsthm}
\usepackage{amssymb}
\usepackage{mathtools}
\usepackage{graphicx}
\usepackage{geometry}
\usepackage[font=footnotesize,labelfont=bf]{caption}
\usepackage[font=footnotesize]{subcaption}
\usepackage{bbm} 
\usepackage{bm} 
\usepackage{mdframed} 
\usepackage{comment}
\usepackage{siunitx}
\usepackage{hyperref}
\usepackage{xurl} 
\usepackage{tikz}
\usetikzlibrary{arrows}
\usetikzlibrary{decorations.pathreplacing}
\usepackage{todonotes}
\usepackage{enumitem}
\usepackage{array}
\usepackage{xcolor}

\usepackage{cancel}
\usepackage{tikz}
\usepackage{color}
\usetikzlibrary{calc,trees,positioning,arrows,chains,shapes.geometric,%
    decorations.pathreplacing,decorations.pathmorphing,shapes,%
    matrix,shapes.symbols}
\hypersetup{
    colorlinks,
    citecolor=black,
    filecolor=black,
    linkcolor=black,
    urlcolor=black
}

\definecolor{DarkerGreen}{RGB}{0,179,45}

\newtheorem{exmp}{Example}[section]

\tikzset{
>=stealth',
  punktchain/.style={
    rectangle, 
    rounded corners, 
    draw=black, very thick,
    text width=10em, 
    minimum height=3em, 
    text centered, 
    on chain},
  line/.style={draw, thick, <-},
  element/.style={
    tape,
    top color=white,
    bottom color=blue!50!black!60!,
    minimum width=8em,
    draw=blue!40!black!90, very thick,
    text width=10em, 
    minimum height=3.5em, 
    text centered, 
    on chain},
  every join/.style={->, thick,shorten >=1pt},
  decoration={brace},
  tuborg/.style={decorate},
  tubnode/.style={midway, right=2pt},
}

\allowdisplaybreaks[3]

\usepackage[toc,page]{appendix} 

\newcounter{for}[section]

\newtheorem{itlemma}{Lemma}[section]
\newtheorem{itproposition}[itlemma]{Proposition}
\newtheorem{itfact}[itlemma]{Fact}
\newtheorem{theorem}[itlemma]{Theorem}
\newtheorem{itcorollary}[itlemma]{Corollary}
\newtheorem{itremark}[itlemma]{Remark}
\newtheorem{itremarks}[itlemma]{Remarks}
\newtheorem{itdefinition}[itlemma]{Definition}
\newtheorem{itexample}[itlemma]{Example}

\newenvironment{fact}{\begin{itfact}\rm}{\end{itfact}}
\newenvironment{claim}{\begin{itclaim}\rm}{\end{itclaim}}
\newenvironment{lemma}{\begin{itlemma}}{\end{itlemma}}
\newenvironment{remark}{\begin{itremark}\rm}{\end{itremark}}
\newenvironment{remarks}{\begin{itremarks} \rm}{\end{itremarks}}
\newenvironment{corollary}{\begin{itcorollary}}{\end{itcorollary}}
\newenvironment{proposition}{\begin{itproposition}}{\end{itproposition}}
\newenvironment{definition}{\begin{itdefinition}\rm}{\end{itdefinition}}
\newenvironment{example}{\begin{itexample}\rm}{\end{itexample}}
\newcommand{\be}[1]{\addtocounter{for}{1} \begin{equation}\label{#1}}
\newcommand{\ee}{\end{equation}}
\newcommand{\bl}[1]{\begin{lemma}\label{#1}}
\newcommand{\br}[1]{\begin{remark}\label{#1}}
\newcommand{\brs}[1]{\begin{remarks}\label{#1}}
\newcommand{\bt}[1]{\begin{theorem}\label{#1}}
\newcommand{\bd}[1]{\begin{definition}\label{#1}}
\newcommand{\bp}[1]{\begin{proposition}\label{#1}}
\newcommand{\bfact}[1]{\begin{fact}\label{#1}}
\newcommand{\bc}[1]{\begin{corollary}\label{#1}}
\newcommand{\bex}[1]{\begin{example}\label{#1}}
\newcommand{\ec}{\end{corollary}}
\newcommand{\efact}{\end{fact}}
\newcommand{\eex}{\end{example}}
\newcommand{\el}{\end{lemma}}
\newcommand{\er}{\end{remark}}
\newcommand{\ers}{\end{remarks}}
\newcommand{\et}{\end{theorem}}
\newcommand{\ed}{\end{definition}}
\newcommand{\ep}{\end{proposition}}
\newcommand{\epr}{\end{proof}}
\newcommand{\bpr}{\begin{proof}}
\newcommand{\bcl}[1]{\begin{claim}\label{#1}}
\newcommand{\ecl}{\end{claim}}

\newcommand{\ecs}{\end{corollary}}
\newcommand{\eers}{\end{exercise}}
\newcommand{\eexs}{\end{example}}
\newcommand{\eems}{\end{example}}
\newcommand{\els}{\end{lemma}}
\newcommand{\eles}{\end{lemmaex}}
\newcommand{\ets}{\end{theorem}}
\newcommand{\eds}{\end{definition}}
\newcommand{\eps}{\end{proposition}}

\newcommand{\bi}{\begin{itemize}}
\newcommand{\ei}{\end{itemize}}
\newcommand{\ben}{\begin{enumerate}}
\newcommand{\een}{\end{enumerate}}

\def\vbar{\mathchoice{\vrule height6.3ptdepth-.5ptwidth.8pt\kern-.8pt}
   {\vrule height6.3ptdepth-.5ptwidth.8pt\kern-.8pt}
   {\vrule height4.1ptdepth-.35ptwidth.6pt\kern-.6pt}
   {\vrule height3.1ptdepth-.25ptwidth.5pt\kern-.5pt}}
\def\fudge{\mathchoice{}{}{\mkern.5mu}{\mkern.8mu}}
\def\bbc#1#2{{\rm \mkern#2mu\vbar\mkern-#2mu#1}}
\def\bbb#1{{\rm I\mkern-3.5mu #1}}
\def\bba#1#2{{\rm #1\mkern-#2mu\fudge #1}}
\def\bb#1{{\count4=`#1 \advance\count4by-64 \ifcase\count4\or\bba A{11.5}\or
   \bbb B\or\bbc C{5}\or\bbb D\or\bbb E\or\bbb F \or\bbc G{5}\or\bbb H\or
   \bbb I\or\bbc J{3}\or\bbb K\or\bbb L \or\bbb M\or\bbb N\or\bbc O{5} \or
   \bbb P\or\bbc Q{5}\or\bbb R\or\bbc S{4.2}\or\bba T{10.5}\or\bbc U{5}\or
   \bba V{12}\or\bba W{16.5}\or\bba X{11}\or\bba Y{11.7}\or\bba Z{7.5}\fi}}

\def \R {{\mathbb R}}

\def \ra {\rightarrow }

\def \a{\alpha}

\def \s{\sigma}

\def \P{{\mathbb{P}}}
\def\g{\gamma}
\def\d{\delta}
\def\e{\varepsilon}

\def\b{\beta}

\def\l{\lambda}
\def\E{{\mathbb{E}}}

\def\1{{\bf 1}}
\def\e{\varepsilon}

\def\ec{\`e }

\def\tZ{\tilde{Z}}
\def\tY{\tilde{Y}}

\title{A stochastic volatility approximation for a tick-by-tick price model with mean-field interaction}

\author{Paolo Dai Pra\footnote{%
Department of Informatics, University of Verona, 
Strada le Grazie 15, 37134 Verona, Italy.
E-mail: \texttt{paolo.daipra@univr.it} 
} 
 \and
Paolo Pigato\footnote{%
Department of Economics and Finance, University of Rome Tor Vergata, Via Columbia 2, 00133 Roma, Italy. 
E-mail: \texttt{paolo.pigato@uniroma2.it
} 
}
}

\date{\today}

\begin{document}

\maketitle

\begin{abstract}
We consider a tick-by-tick model of price formation, in which buy and sell orders are modeled as self-exciting point processes (Hawkes process), similar to the one in  [Bacry, Delattre, Hoffmann, Muzy, \emph{Modelling microstructure noise with mutually
exciting point processes}, Quantitative Finance, 2013] and [El Euch, Fukasawa, Rosenbaum, \emph{The microstructural foundations of leverage effect and rough volatility}, Finance and Stochastics, 2018]. We adopt an agent based approach by studying the aggregation of a large number of these point processes, mutually interacting in a mean-field sense. 

The financial interpretation of the model is that of an asset on which several labeled agents place buy and sell orders following these point processes, influencing the price. The mean-field interaction introduces positive correlations between order volumes coming from different agents that reflect features of real markets such as herd behavior and contagion. 
When the large scale limit of the aggregated asset price is computed, if parameters are set to a {critical} value, a singular phenomenon occurs: the aggregated model converges to a stochastic volatility model with leverage effect and faster-than-linear mean reversion of the volatility process. 

The faster-than-linear mean reversion of the volatility process is supported by econometric evidence, and we have linked it in [Dai Pra, Pigato, \emph{Multi-scaling of moments in stochastic volatility models},
Stochastic Processes and their Applications, 2015] to the observed multifractal behavior of assets prices and market indices. This seems connected to the Statistical Physics perspective that expects anomalous scaling properties to arise in the critical regime.  
\end{abstract}

\bigskip

\noindent{\textbf{Keywords: }}
Stochastic Volatility, Hawkes processes, multifractality, mean-field, non-linearity, criticality

\medskip

\noindent{\textbf{AMS 2020: }} 
	primary: 
			60G55,  	
    	secondary: 	60F05,  	
    60G44,  	
91G45  	

\section{Introduction}

We consider a tick-by-tick model of price formation, 
in which 
price variations are due to buy and sell orders
 of individual agents, that 
are modeled as self-exciting point processes (Hawkes process),  and are mutually interacting in the mean-field sense.
 Our main aim is to use this model to provide a microscopic foundation to stochastic volatility models in which the mean reversion of the volatility is {\em faster-than-linear}. Supported by econometric evidence \cite{BAKSHI}, models with quadratic mean reversion in the volatility process have been used for option pricing (e.g. on bitcoins) and connected issues  \cite{sepp,lewis}, and it is the drift form in the volatility function of the $3/2$-model \cite{baldeaux}. In \cite{daipra.pigato}, we have linked faster-than-linear mean reversion of the volatility process to the observed multifractal behavior of asset and index prices \cite{Jiang_2019}. 

In this paper we show that the price process, defined as the difference between the total numbers of buy and sell orders, suitably rescaled with the number of agents, converges in distribution to the price process of a stochastic volatility model with quadratic mean reversion in the volatility. In what follows we illustrate further details of our model and results, in comparison with the existing literature.

We follow here the general principle, inspired by Statistical Physics, of investigating to what extent macroscopic financial dynamics can be derived as limits of microscopic, possibly agent-based models \cite{Chakraborti01072011,Cont_Bouchaud_2000}. Moreover, inspired in particular by 
\cite{Bacry01012013,el2018microstructural,jaisson.rosenbaum,bacry1,viens}, the microscopic dynamics of buy and sell orders is modeled with Hawkes processes, which are self-exciting point processes. In particular, in the recent work \cite{el2018microstructural}, a rough version of the Heston model is derived as high-frequency limit of a market model in which price variations due to buy and sell orders are modeled as Hawkes. In addition to roughness, this work provides a microscopic foundation of the \emph{leverage effect}. Unlike \cite{el2018microstructural}, where the microscopic model concerns the total flow of buy and sell orders, we introduce an agent based dynamics, by modeling the orders of each individual agent. The intensity according to which agents place buy and sell orders is a function of a time-weighted integral of the total past orders; this introduces a simple form of mean field interaction among agents. The simplicity lies in the fact that the total numbers of buy and sell orders are {\em sufficient statistics}: their aggregated dynamics can be expressed without referring to individual agents. This dimensional reduction appears in other classical mean field models, such as the mean field versions of the Voter Model, the Contact Process and the stochastic Ising model. As a consequence, the total numbers of buy and sell orders follows dynamical rules similar to those in \cite{el2018microstructural}, with the difference that we allow a nonlinear dependence of the intensities on the past orders. Besides this nonlinearity, that we discuss later, the agent based formulation offers at least two advantages. On one hand the hydrodynamic scaling needed for the limit is clearly linked to the size of agent's population, providing a sound interpretation of the scaling regime to which the limit stochastic volatility model provides a good approximation to the microscopic dynamics. On the other hand, our model allows natural extensions, including  inhomogeneities in the agent's population (e.g. {\em informed} and {\em uninformed agents}) or metaorders modeling through high self-excitation of orders by the same agent; in the last case the aggregated numbers of buy and sell orders do not admit closed dynamical roles. These extensions will be discussed in Subsection \ref{subsec:ext}.

The nonlinearity in the dependence of the intensities on the past orders is the main feature of the model we propose.  The linearity assumption is used e.g. in \cite{el2018microstructural}, and it is a key ingredient in the asymptotic analysis obtained there. We show that the introduction of nonlinearity, which can be motivated by saturation effects, significantly changes the asymptotic behavior. The limit is still a stochastic volatility models, but with nonlinear mean reversion. Models with superlinear mean reversion have been proposed and supported by several authors \cite{BAKSHI,sepp,lewis,baldeaux,daipra.pigato}; to our knowledge this is the first research proposing a justification in the context of price dynamics.
Relaxing linearity, moreover, forces us to use different arguments form those in \cite{el2018microstructural}. We make a strict assumption on the memory function, that we suppose exponential; this allows a Markovian description of the pre-limit process and the use of tools from weak convergence  to diffusion processes. Nonlinearity reveals the possibility of scaling properties which differ from the linear case, and naturally leads to quadratic mean reversion. This link between nonlinearity of the intensities in the pre-limit process and nonlinearity of the volatility mean reversion reveals further the effectiveness and flexibility of modeling with Hawkes processes. 

Similarly to \cite{el2018microstructural,szymanski,Bouchaud}, we work under the assumption of {\em near instability}: model's parameters are set on the boundary of the stability region, i.e. the region where (rescaled) intensities remain bounded over time. In a financial setting, this represents the fact that financial markets operate in a regime of near-critical endogeneity, where the majority of market events are generated by feedback from prior events, in the regime where even a slight increase in feedback would render the system unstable \cite{Hardiman2013Critical,HardimanBouchaud2014,Wehrli2020}.

This corresponds to the same principles that led to systematic study of {\em critical} models in statistical mechanics. The question of why many real complex systems self-organize at or close to the critical point is, to a large extent, yet to be understood at a rigorous level. The notion of {self-organized criticality} was proposed in the fundamental paper \cite{bak88}. In few models this phenomenon has allowed a rigorous analysis; we mention the sandpile models  \cite{AJ04} and Curie-Weiss type models \cite{CG16}.
Self organized criticality is often associated with the existence of multiple scaling regimes (multifractality), see \cite{MG14} for an account. Multifractality is a well extablished stylized fact of market indices,
stocks, commodities, exchange rates, interest rates and other financial time series \cite{Jiang_2019}. In connection to these facts, in \cite{daipra.pigato}, we showed that in a certain class of Levy-driven volatility models, multifractality (in the form of multiscaling of moments, see next Remark \ref{rm:multiscaling}) is possible only if the mean-reverting drift function is superlinear. 
We remark that, in the model we analyse in the present paper, the nonlinearity of the mean reversion is an example of {\em anomalous} critical scaling. If the parameters were set below the critical point, then the price process would converge to a model with linear mean reversion, but that would become trivial as  the critical point is approached, in the sense that the mean reversion would vanish. At the critical point the time needs to be rescaled to obtain a nontrivial limit, and linearity is lost in the limit.  This is close in spirit to anomalous fluctuation theorems in models motivated by Statistical Mechanics, where non Gaussian distributions appear in the limit \cite{CE88, BS07}. We further mention \cite{sornette2013critical} for a comprehensive review on this topic, and \cite{challet.marsili} for applications to market models.

The nonlinearity of the model, and the related anomalous scaling, have a strong impact on the techniques needed to establish the final limit theorem. For a brief illustration of this point, we return to the fact that, due to the exponential memory function, the intensity of pre-limit process admits a Markovian description. Thus it could be expected that, at least at a semi-rigorous level, the limit theorem could be derived from strong limits of infinitesimal generators. This is not the case, as the infinitesimal generator has terms that diverge as the number of agents goes to infinity. These terms have the effect of making the dynamics to collapse onto a low dimensional manifold. This type of problem appears in perturbation theory of operators (see \cite{da83} for an application to a mean field model).

Finally, we mention again that our model keeps a relevant property of the linear models in \cite{el2018microstructural}, namely the leverage effect, which refers to the widely documented negative correlation between observed changes in stock returns and volatility \cite{engle.ng}, see also \cite{el2018microstructural,lp} and references therein. 
It is often incorporated in models for financial price dynamics through a negative correlation between the noises driving price and volatility, see \cite{hagan,heston,hull.white,stein}.
Leverage effect can be explained by the fact that when an asset price decreases, the ratio of the company's debt with respect to the equity value becomes larger, and as a consequence volatility increases. Other ``macroscopic'' explanations from financial economics have been proposed \cite{black1976,Christie:1982vf,Hentschel,French}, as well as ``psycological'' explanations based on asymmetric responses by investors \cite{Hens:2006tv}. Here, as in \cite{el2018microstructural}, leverage effect observed on macroscopic quantities is a product of microstructural features, with in addition an explicit, stylized representation of interactions between agents.

\paragraph{Acknowledgements.}
We are grateful to Mathieu Rosenbaum, Fabrizio Lillo and Alessandro Calvia for discussion and several suggestions that led to an improved version of the paper. We also thank the editor and reviewer for their careful reading and useful insights. 
PP acknowledges financial support from University of Rome Tor Vergata via Project E83C25000470005.

\paragraph{Outline.} 
In Section \ref{sec:results} we describe our model, state the main results of the paper and discuss them. Section \ref{sec:proofs} contains proofs and technical material, while in the appendix we collect known results on which our proofs hinge and some auxiliary computations.

\section{Model description and main results}\label{sec:results}

We consider a market where $N$ agents are placing  orders for a given asset.  For each agent $i$ we let $N_i^+$ (resp. $N_i^-$) be the processes counting her buy (resp. sell) orders. We assume the contribution of agent $i$ to the asset log-price is  the sum of the positive and negative orders placed, namely
\be{def:singlep}
P_i = N^+_i - N^-_i.
\ee
We refer the reader to 
\cite{ farmer2004what, Lillo2023OrderFlow, Gatheral01082010,Huberman,durin.et.al} and references therein 
for a discussion of the process of price formation in financial markets as a consequence of the arrival of orders, and
for more information on impact of buy and sell orders on price. We also interpret this model as a decomposition over individual agents of the model in \cite{el2018microstructural,Bacry01012013}. See next Remark \ref{rm:agg.model} for an alternative interpretation in the case of market indices or portfolios.

We now describe the dynamics of the processes $N_i^{\pm}$. 
Consider the function $\varphi(t) = e^{-\a t}$, ruling the memory of the  process. Note that more general choices of $\varphi:[0,+\infty) \ra [0,+\infty)$ have been considered, see for example \cite{el2018microstructural}; however, our analysis relies on this choice. Let us introduce two sequences $\nu_N^{\pm}$ of $\s$-finite measures on $[0,+\infty)$,  and define
\begin{equation}\label{def:X}
X^{\pm}_i(t) := \int_{[0,t]} \varphi(t-s) \left[dN^{\pm}_i(s) + \nu_N^{\pm}(ds)\right].
\end{equation}
Note that $X_i^{\pm}(t) \geq 0$. The term $\int_{[0,t]} \varphi(t-s) dN^{\pm}_i(s)$ is the sum of the past jumps of the process $N_i^{\pm}$ weighted by the memory function $\varphi$. The (positive) measures $\nu_N^{\pm}$ represent external signals, and impact all agents; precise assumptions on $\nu_N^{\pm}$ will appear later.
We then consider the empirical means
\[
m_N^{\pm}(t) := \frac{1}{N} \sum_{i=1}^N X^{\pm}_i(t),
\]
that, as we will see in Lemma \ref{lemma:mN}, form a two dimensional Markov process.
We will assume that the point processes $(N^{\pm}_i)_{i=1}^N$ are conditionally independent given $(X^{\pm}_i)_{i=1}^N$. Each ``upward'' jump process $N^+_i$ has a stochastic intensity given by
\begin{equation}\label{def:lambda}
\l_N^+(t) := f(m_N^+(t) + \b \g m_N^-(t)),
\end{equation}
while the``downward" jumps have stochastic intensity
\begin{equation}\label{def:lambdameno}
\l_N^-(t) := f( \g m_N^+(t) + (1+(\b-1) \g) m_N^-(t)),
\end{equation}
where 
$f$ is a given increasing, concave function of class $\mathcal{C}^3$, with $f(0) = 0$, $f'(0)>0$, $f''(0)<0$ with $f'$, $f''$ and $f'''$ bounded,  and $\b \geq 1,  \, \g\in [0,1]$ are parameters ruling the self and cross excitation of upward and downward jumps (see Remarks \ref{rm.jump.parameters} and \ref{rm.f}).
 
The processes $(N^\pm_i)_{i=1}^n$  are called Hawkes processes. Note that the intensity is an increasing function of the mean $m_{N}^{\pm}$ and is the same for each agent, which should reflect contagion and homogeneity of the model. The \emph{mean-field} interaction models the fact that buy orders, with consequent increase in the log-price (and viceversa, sell orders with  consequent decrease in the log-price) are exciting orders of the same type in other labeled agents.
Moreover, the processes $X_i^{\pm}$ that are responsible for the self-excitation of the $N_i^{\pm}$, are also driven by a common, external signal modeled by the measure $\nu_N^{\pm}$, representing a ``baseline'' volatility. In what follows we will assume that
\[
\nu_N^{\pm}(dt) = a^{\pm}_N \delta_0(dt) + b_N^{\pm} dt,
\]
where $\delta_0$ denotes the Dirac measure and $a_N^{\pm}, b_N^{\pm}$ are positive constants satisfying
\be{assab}
\lim_{N \ra +\infty} \sqrt{N} a_N^{\pm} = a^{\pm} \in (0,+\infty), \ \ \ \lim_{N \ra +\infty} \sqrt{N} b_N^{\pm} = b^{\pm} \in (0,+\infty).
\ee
Note that  $a_N^{\pm}$ is responsible for the value at time $t=0$ of the stochastic intensities $\l_N^{\pm}$, while  $b_N^{\pm}$ represent stationary external signals. 

We now introduce the rescaled log-price process
\be{def:pn}
\Pi_N(t)  = \frac{1}{\sqrt{N}} \sum_{i=1}^N P_i(\sqrt{N}t). 
\ee
Our main result is the following. 

\begin{theorem} \label{th:main}
Suppose the criticality condition 
\be{intro:critical}
\a = f'(0)(1+\b\g)
\ee
holds.
The process $\Pi_N$ converges in  distribution in $(0,T]$ to the process $\pi$, where $(\pi,y)$ solves the SDE
\begin{equation} \label{volmod}
\begin{split}
d\pi(t) & =\beta_{\pi} dt + \sigma_{\pi} \sqrt{f'(0) y(t)} dW(t) \\
dy(t)  & = (\beta_{y}  + \alpha_{y} f''(0) y^2(t) )dt + \sigma_{y} \sqrt{f'(0)y(t)} dB(t)
\end{split}
\end{equation}
with initial condition
\[
\pi(0)  = \frac{1-\g}{\g(1+\b)}(a^+ - a^-),\quad 
y(0)  = \frac{1+\b\g}{1+\b} a^+ + \b \frac{1+\b\g}{1+\b} a^-,
\]
where the constants in the dynamics depend on $\beta,\gamma,b^{\pm}$ as below
\[
\begin{split}
\beta_{\pi}&= \frac{1-\g}{\g(1+\b)}(b^+ - b^-);\quad
\sigma_{\pi}=\sqrt{2} \frac{1+\b\g}{\g(1+\b)} ;\\
\beta_{y}&=\frac{1+\b\g}{1+\b} b^+
+ \b \frac{1+\b\g}{1+\b} b^- ;\quad
\alpha_{y}= \frac12 (1+\b\g);\quad
\sigma_{y}= \frac{\sqrt{(1+\b^2)}(1+\b\g)}{1+\b}
\end{split}\]
and $W,B$ are standard Brownian motions, with correlation
\be{intro:limcorr}
d\langle B,W \rangle_t =  \frac{(1-\b^2)\g}{(1+\b\g)\sqrt{2(1+\b^2)}} dt. 
\ee
\label{th:initial}
\end{theorem}

\begin{remark}\label{rm.jump.parameters}
Note that \eqref{def:lambda} and \eqref{def:lambdameno} express the jump intensity as a concave function of a functional of past jumps. For instance, referring to \eqref{def:lambda}, the term $m^+_N(t)$ can be interpreted as the impact of past jump upwards on $\lambda^+_N$, and $\beta \gamma m_N^-(t)$ as the impact of past jump downwards on $\lambda^+_N$. The condition $\gamma \leq 1$ corresponds to the requirement ``past jump upwards (resp. downwards) impact more on $\lambda^+_N$ than on $\lambda^-_N$ (resp. more on $\lambda^-_N$ than on $\lambda^+_N$)'', which we need to encode herd behavior in our model. Indeed,
\[
m_N^+ \geq \g m_N^+  \mbox{ and } (1+(\b-1)\g)m_N^- \geq \b \g m_N^-,
\]
with equality for $\g = 1$. The condition $\b \geq 1$ is related to the higher impact of the jump downwards, in the following sense:
``the impact of jump downwards on $\l_N^+$ (resp. $\l_N^-$)
is greater than the impact of jump upwards on 
$\l_N^-$ (resp. $\l_N^+$). 
Indeed the coefficient $\b\g$ of $m_N^-(t)$ in \eqref{def:lambda} is greater than the coefficient $\g$ of $m_N^+(t)$ in \eqref{def:lambdameno}, and the coefficient $(1+(\b-1)\g)$ of $m_N^-$ in \eqref{def:lambdameno} is greater than the coefficient $1$ of $m_N^+$ in \eqref{def:lambda}, with equality for $\b=1$. 
This condition seems natural in a financial context due to the higher ``excitatory power'' of downard price movements, i.e., the fact that investors react more to a decrease than to an increase in the price; cf. also the explanation in \cite{el2018microstructural}. This  corresponds to $\beta\geq 1$, and gives rise to leverage effect in the limit system in the sense of Remark \ref{rm.leverage}.
Note also that the requirements on the impact of jump upwards and downwards lead naturally to introduce at least two parameters in the model, possibly with a different combination from the one we have chosen.
\end{remark}

\begin{remark}\label{rm.f}
In our setting, equations \eqref{def:lambda} and \eqref{def:lambdameno} encompass the fact that the effect on present jumps intensity of compounded past jumps averaged over all the agents is not linear, but mediated by the increasing concave function $f$,  meaning that a saturation effect is at play, with larger past jumps producing a stronger intensity of present jumps, with a decreasing marginal dependence. Note that the $y$ process in \eqref{volmod} is mean reverting, since $f''(0)<0$.
\end{remark}

\begin{remark}[Leverage effect]\label{rm.leverage}
Since $\gamma\in[0,1]$ and $\beta\geq 1$, it always holds that ${(1+\b)\g}/{(1+\b\g)}\in [0,1]$ and $\langle B,W \rangle\in(-1/\sqrt{2},0]$. When $\b\to\infty$, we have $\langle B,W \rangle\to -1/\sqrt{2}$. Recall that in \cite{el2018microstructural}, the correlation in the limit model is ${(1-\b)}/{\sqrt{2(1+\b^2)}}$, which we obtain here with $\gamma=1$. The fact that  $\langle B,W \rangle\leq 0$ means that this model displays a negative relationship between changes in stock returns and in volatility, the so called leverage effect.
This is a reflection of the fact that $\beta\geq 1$, representing, as seen in Remark \ref{rm.jump.parameters}, the fact that  jumps downwards are overall more excitatory than jumps upwards.
\end{remark}

\begin{remark}\label{rm:agg.model}
[Index and portfolio dynamics]
A different possible interpretation of the aggregated price in \eqref{def:pn} and its limit  $(\pi_{t})_{t}$, is that of an index or portfolio composed of $i=1,\dots, N$ assets. The variations of the index are given by an average of the variations of the single components of the index, each represented by the sum of its positive and negative ``micro'' variations. Each  individual asset has dynamics as in \cite{el2018microstructural,Bacry01012013}. The asset prices are self exciting and are exciting each other, with jumps in one price exciting jumps in all the others through a mean-field interaction. The diffusive dynamics is seen at a slower time scale than the order book time scale where we see the jumps in the prices.
\end{remark}

\begin{remark}[Criticality as a stylized fact of financial markets] Empirical research strongly suggests that markets self-organize so as to be poised at the border between stability and instability. In particular, financial markets operate in a regime of near-critical endogeneity, where the majority of market events are generated by feedback from prior events, and the system hovers at the boundary of stability, such that even a slight increase in feedback would render the system unstable.

In \cite{Hardiman2013Critical}, high-frequency futures data are analysed using Hawkes processes, and the branching ratio (the expected number of “offspring” events triggered by a single trade) is consistently found to be close to one, indicating that market activity is largely endogenous. Other studies confirm that calibrated branching ratios across different assets and time periods remain close to unity \cite{HardimanBouchaud2014,Wehrli2020}, supporting the idea that markets are self-organized near this critical threshold. These findings indicate that financial markets operate in a marginally stable, highly endogenous regime, where feedback dominates dynamics, in line with the nearly-unstable Hawkes framework of \cite{el2018microstructural}.

We interpret our criticality condition \eqref{intro:critical} as encoding an analogous fact, in an interacting setting. Specifically, the memory parameter $\alpha$ controls the level of endogeneity, i.e., the persistence of the influence of previous jumps on present activity. Larger values of $\alpha$ correspond to shorter memory and reduced endogeneity, while smaller values of $\alpha$ correspond to longer memory and stronger endogeneity. Condition \eqref{intro:critical} states that the memory of the system sits exactly at the critical point where any further increase in feedback would break stability in the limit system in Theorem \ref{th:main}.
\end{remark}

\begin{remark}[Multifractality - Multiscaling of moments]
Let $(X_t)_{t \geq 0}$ be a continuous-time martingale, having stationary increments; in financial applications this could be identified with the de-trended {\em log-price} of an asset, or the log-price with respect to the martingale measure used for pricing derivatives. We say that multiscaling of moments occurs if $\E\left(|X_{t+h} - X_t|^q \right)$ scales, in the limit as  $h \downarrow 0$, as $h^{A(q)}$, with $A(q) = \frac{q}{2}$ for small $q$ but $A(q) < \frac{q}{2}$ beyond a certain threshold.  The presence of multiscaling in financial time series is well established and various models have been proposed to capture it, most notably multi-fractal processes \cite{calvet.fisher}, see also \cite{wu.muzy.bacry,NR18,acdp,BRANDI}. In \cite{daipra.pigato} we have shown that in a stochastic volatility model whose volatility $V_t$ solves the stochastic differential equation
\be{sde}
dV_t \, = \, - f(V_t) dt + dL_t,
\ee
with $L$ a {\em Levy process},  multiscaling occurs if and only if the characteristic measure of $L$ has power law tails at infinity, and $f$ is a {\em superlinear} mean reversion function, a fact that is supported e.g. by the empirical findings in \cite{BAKSHI}.  Therefore, the result in Theorem \ref{th:initial} is a possible explanation of the origin of the superlinear mean reversion, while jumps may be produced by exogenous shocks.
\label{rm:multiscaling}

\end{remark}

\subsection{Extensions of the model} \label{subsec:ext}

In this section we illustrate two extensions of the model in Theorem \ref{th:main}. First we break the symmetry among agents, letting the model's parameters to vary within the population. Later we introduce self-interaction: agent's history has more impact on her own future actions than on that of other players, which we interpret as the fact that large trading orders are often split into smaller orders, executed over a certain time span, leading to an ``apparent'' high self-excitation \cite{moro2009marketimpact}.
This last model loses one of the key properties of the former two; namely that the aggregated means $m_N^{\pm}(t)$ evolve as a two dimensional Markov process  (see Lemma \ref{lemma:mN}). We state the theorems concerning the limiting behavior of these models. The proofs require rather mild modifications with respect to that of Theorem \ref{th:main}, and are given in Section \ref{sec:proofs}.

\bigskip

\noindent
{\bf Inhomogeneous agents}
\smallskip

The agent based formulation allows to introduce heterogeneities in the population of agents, in particular in their sensitivity to past jumps. By this we mean that the stochastic intensities $\l_N^{i,\pm}$ depend on the individual agent:
\[
\begin{split}
\l_N^{i,+}(t) &  := f(m_N^+(t) + \b^N_i \g^N_i m_N^-(t)) \\
\l_N^{i,-}(t) &  := f(\g^N_i m_N^+(t) + (1+(\b^N_i-1) \g^N_i )m_N^-(t)).
\end{split}
\]
Note that agents react to the same signals $m_N^{\pm}$, but their ''sensitivity'' may differ: faster reaction to signals could model, for instance, higher degree of information. The fact of allowing the parameters $\b^N_i \geq 1$ and $\g^N_i \in [0,1] $ to depend on both $i$ and $N$ allows a simpler formulation of the needed assumptions. For the same reason, also the parameter $\a = \a_N$ of the memory function $\varphi(t) = e^{- \a_N t}$ is allowed to depend on $N$. To formulate the needed condition on the parameters we introduce the {\em empirical measure}
\[
\rho_N := \frac{1}{N} \sum_{i=1}^N \delta_{(\b^N_i,\g^N_i)},
\]
and make the following assumptions:
\be{ass:empdis} 
\begin{array}{c}
\rho_N \mbox{ converges in distribution to } \rho \\ \\
\lim_{N \ra +\infty} \a_N = \a,
\end{array}
\ee
where $\rho$ is a probability on $[1,+\infty)\times [0,1]$, with compact support, and $\a>0$. We also introduce the following notation: for a function $g: [1,+\infty)\times [0,1] \rightarrow \R$ set
\[
\overline{g(\b,\g)}_N := \frac{1}{N} \sum_{i=1}^N g(\b^N_i,\g^N_i), \hspace{1cm}  \overline{g(\b,\g)} := \int g(\b,\g)\rho(d\b,d\g).
\]
Note that 
\[
\overline{g(\b,\g)} = \lim_{N \ra +\infty} \overline{g(\b,\g)}_N
\]
if $g$ is continuous. We are now ready to state the extension of Theorem \ref{th:main}.

\begin{theorem} \label{th:maininhom}
Suppose the criticality condition 
\be{}
\a_N = f'(0)(1+\overline{\b\g}_N)
\ee
holds for every $N$.
The process $\Pi_N$ converges in  distribution in $(0,T]$ to the process $\pi$, where $(\pi,y)$ solves the SDE
\begin{equation} \label{volmodinhom}
\begin{split}
d\pi(t) & =\beta_{\pi} dt + \sigma_{\pi} \sqrt{f'(0) y(t)} dW(t) \\
dy(t)  & = (\beta_{y}  + \alpha_{y} f''(0) y^2(t) )dt + \sigma_{y} \sqrt{f'(0)y(t)} dB(t)
\end{split}
\end{equation}
with initial condition
\[
\pi(0)  = \frac{1-\overline{\g}}{\overline{\g} + \overline{\beta\g}}(a^+ - a^-),\quad 
y(0)  = \frac{\overline{\g}(1+\overline{\b\g})}{\overline{\g}+\overline{\b\g}} a^+ + \frac{\overline{\b\g}(1+\overline{\b\g})}{\overline{\g}+\overline{\b\g}}a^-,
\]
where the constants in the dynamics depend on $\rho$ and $b^{\pm}$ as below
\[
\begin{split}
\beta_{\pi}&= \frac{1-\overline{\g}}{\overline{\g} + \overline{\b\g}}(b^+ - b^-);\quad
\sigma_{\pi}=\sqrt{2} \frac{1+\overline{\b\g}}{\overline{\g} + \overline{\b\g}} ;\\
\beta_{y}&=\frac{\overline{\g}(1+\overline{\b\g})}{\overline{\g}+\overline{\b\g}} b^+ + \frac{\overline{\b\g}(1+\overline{\b\g})}{\overline{\g}+\overline{\b\g}}b^- ;\quad
\alpha_{y}= \frac12 \frac{\overline{\g} \overline{(1+\b\g)^2}}{(\overline{\g}+\overline{\b\g})(1+\overline{\b\g})};\quad
\sigma_{y}= \sqrt{\overline{\g}^2 + \overline{\b\g}^2} \frac{1+\overline{\b\g}}{\overline{\g} + \overline{\b\g}}
\end{split}
\]
and $W,B$ are standard Brownian motions, with correlation
\be{intro:limcorrinhom}
d\langle B,W \rangle_t =  \frac{\overline{\g}^2 - \overline{\b\g}^2}{(1+\overline{\b\g})\sqrt{2(\overline{\g}^2+\overline{\b\g}^2)}} dt. 
\ee
\label{th:initialinhom}
\end{theorem}

\medskip
\noindent
{\bf Self excitation}

\smallskip

The intensities defined in \eqref{def:lambda} and \eqref{def:lambdameno} are the same for all agents, and only depend on the total number of jumps. 
We assume in this extension that an individual player weights her own history more than that of other agents, representing the apparent effect of many small orders, clustered in time, coming from the same agent, resulting from the execution of a large metaorder. A possible way of modeling this phenomenon is to replace the intensities in \eqref{def:lambda} and \eqref{def:lambdameno} by the agent dependent intensities
\be{lambdai}
\begin{split}
\l_N^{i,+}(t) &  := f(m_N^{i,+}(t) + \b \g m_N^{i,-}(t)) \\
\l_N^{i,-}(t) & := f( \g m_N^{i,+}(t) + (1+(\b-1) \g) m_N^{i,-}(t)),
\end{split}
\ee
where
\[
m_N^{i,\pm}(t) = m_N^{\pm}(t) + \frac{\kappa}{\sqrt{N}} X^{\pm}_i(t)
\]
for a positive constant $\kappa$. In this way the history of agent $i$ has an impact of order $\frac{1}{\sqrt{N}}$, while that of other players have impact of order $\frac{1}{N}$. Cleary many other scaling could be proposed; we choose this as it leads to a slight modification of Theorem \ref{th:main}. For this model, as soon as $f$ is nonlinear, the dynamics of aggregated variables $m_N^{\pm}$ is non Markovian (see Remark \ref{rem:nonmar}). However, the Markov property is restored at the limit, which is the same as in Theorem \ref{th:main}, except for an extra linear term in the drift of the volatility.

\begin{theorem} \label{th:mainself}
Suppose the criticality condition 
\be{intro:criticalself}
\a = f'(0)(1+\b\g)
\ee
holds.
The process $\Pi_N$ converges in  distribution in $(0,T]$ to the process $\pi$, where $(\pi,y)$ solves the SDE
\begin{equation} \label{volmodself}
\begin{split}
d\pi(t) & =\beta_{\pi} dt + \sigma_{\pi} \sqrt{f'(0) y(t)} dW(t) \\
dy(t)  & = (\beta_{y} +  \theta_y y(t)+ \alpha_{y} f''(0) y^2(t) )dt + \sigma_{y} \sqrt{f'(0)y(t)} dB(t)
\end{split}
\end{equation}
where
\[
\theta_y = f'(0)\frac{1+\b\g}{1+\b}\kappa,
\] 
and all remaining parameters and initial conditions are as in Theorem \ref{th:main}.
\end{theorem}

\section{Proofs}\label{sec:proofs}
The proof of Theorem \ref{th:main} is divided in  two steps. We first identify a {\em microscopic } volatility process $Y_N(t)$, and we show it converges in distribution to a solution of the second equation in \eqref{volmod}. We use martingale methods for this convergence; the main difficulty is to show that certain terms in the dynamics of $Y_N$ vanish as $N \rightarrow +\infty$, which is essential to obtain a closed equation for the limit of $Y_N$. As it is customary in critical dynamics, there is one divergent term (as $N \rightarrow +\infty$) in the dynamics of $Y_N$, which however vanishes if the criticality condition \eqref{intro:critical} holds. In the second step we prove the convergence of the price process. This turns out to be more difficult, as the divergent term in the price dynamics does not disappear at criticality. This type of problem appears in perturbation theory of operators (see \cite{da83} for an application to a mean field model). Here we circumvent the problem by applying a theorem on {\em collapsing processes} due to Comets \& Eisele \cite{CE88}, which may be seen as the probabilistic counterpart of a perturbation theory argument.
Theorem \ref{thm:diffusion:approx}, Theorem \ref{collapsing}, Lemma \ref{lemma:collapsing} and Theorem \ref{thm:eq.quad}, which we use along the proof, are given in the appendix. 

In this section we denote by $A,B,C$ constant independent of $N$, that may vary from line to line.

\subsection{Preliminaries}

We begin by giving a semimartingale representation for the empirical means.
\[
m_N^{\pm}(t) := \frac{1}{N} \sum_{i=1}^N X^{\pm}_i(t),
\]
in terms of a family $n_i^+(ds,du)$, $n_i^-(ds,du)$ of independent Poisson Random Measures on $[0,+\infty)^2$ with intensity measure $ds \, du$.
\begin{lemma}  \label{lemma:mN}
The following semimartingale representations hold:
\begin{equation} \label{macsde}
\begin{split}
m_N^{\pm}(t) & = - \a \int_0^t m_N^{\pm}(s) ds + \int_{[0,t]\times [0,+\infty)} {\bf 1}_{[0,\l_N^{\pm}(s-))}(u) \frac{1}{N} \sum_{i=1}^N n_i^{\pm}(ds,du)+ a^{\pm}_N + b^{\pm}_N t \\ & = - \a \int_0^t m_N^{\pm}(s) ds  + \int_0^t \l_N^{\pm}(s)ds + a^{\pm}_N + b^{\pm}_N t+ M_N^{\pm}(t),
\end{split}
\end{equation}
where $\l_N^{\pm}$ are given in \eqref{def:lambda} and \eqref{def:lambdameno} and, letting $\tilde{n}_i^{\pm}(ds,du) = n_i^{\pm}(ds,du) - ds \, du$, 
\[
M_N^{\pm}(t) := \int_{[0,t]\times [0,+\infty)} {\bf 1}_{[0,\l_N^{\pm}(s-))}(u) \frac{1}{N} \sum_{i=1}^N \tilde{n}_i^{\pm}(ds,du)
\]
are orthogonal martingales with  conditional quadratic variations $\frac{1}{N} \int_0^t \l_N^{\pm}(s)ds$  (see \cite[Chapter III, Section 5, page 124]{Protter2005} for the definition of conditional quadratic variation). Moreover the processes $m_N^{\pm}(t)$ form a two dimensional Markov process.
\end{lemma}
\begin{proof}

Using our specific choice $\varphi(t)=e^{-\alpha t}$ , by Fubini's Theorem and the fact that $\varphi(t) = -\frac{1}{\a} \frac{d}{dt} \varphi(t)$:
\[
\begin{split}
\int_0^t X^{\pm}_i(s) ds & =  \int_0^t \int_{[0,s]} \varphi(s-r)\left[ dN^{\pm}_i(r) +\nu_N^{\pm}(dr)\right]ds = \int_{[0,t]} \left(\int_r^t \varphi(s-r) ds \right) \left[ dN^{\pm}_i(r) + \nu_N^{\pm}(dr)\right] \\
& =  \int_{[0,t]}\frac{1}{\a} \left(1-\varphi(t-r)\right) \left[ dN^{\pm}_i(r) + \nu_N^{\pm}(dr)\right]  = \frac{1}{\a}N^{\pm}_i(t) + \frac{1}{\a} \nu_N^{\pm}([0,t]) - \frac{1}{\a} X^{\pm}_i(t),
\end{split}
\]
and therefore
\begin{equation} \label{Markov}
X^{\pm}_i(t) =  - \a \int_0^t X^{\pm}_i(s) ds  + N^{\pm}_i(t) + \nu_N^{\pm}([0,t]).
\end{equation}
Note that this computation relies on our specific functional choice for $\varphi$, and not only on its asymptotics.

Considering a family $n_i^+(ds,du)$, $n_i^-(ds,du)$ of independent Poisson Random Measures on $[0,+\infty)^2$ with intensity measure $ds \, du$, \eqref{def:X} can be rewritten as
\begin{equation} \label{micsde}
X^{\pm}_i(t) = - \a \int_0^t X^{\pm}_i(s) ds + \int_{[0,t]\times [0,+\infty)} {\bf 1}_{[0,\l_N^{\pm}(s-))}(u) n_i^{\pm}(ds,du)+ a^{\pm}_N + b^{\pm}_N t,
\end{equation}
which can be aggregated to obtain \eqref{macsde}. 
 Finally, the Markov property of $m_N^{\pm}(t)$ follows from the representation \eqref{macsde}, the independence properties of Poisson Random Measures and the fact that the processes $\l_N^{\pm}(t)$ are deterministic functions of $m_N^{\pm}(t)$ (see \eqref{def:lambda} and \eqref{def:lambdameno}).
\end{proof}
\begin{remark}\label{rem:sde}
Note that equations \eqref{micsde} show that the trajectories of $X^{\pm}_i$ are uniquely determined by the realization of the points of the Poisson Random Measures $n_i^{\pm}$. Indeed, $X_i^{\pm}(t)$ evolves deterministically until the moving intervals $[0,\l_N^{\pm}(t))$ meet a point of one of the $n_i^{\pm}$, which causes a unit jump of $X_i^{\pm}(t)$.
\end{remark}

\subsection{Definition and convergence of the volatility process}
We now define
\[
Y_N(t) 
 = \sqrt{N} \left[\frac{1+\b\g}{1+\b} m_N^+(\sqrt{N}t) + \frac{\b(1+\b\g)}{1+\b} m_N^-(\sqrt{N}t)\right]
\]
and
\[
Z_N(t) 
 = \frac{(1-\g)\sqrt{N}}{2} \left[  m_N^+(\sqrt{N}t) -  m_N^-(\sqrt{N}t)\right].
\]
Note that
\[
\l_N^{+}(t) = f \left( \frac{1}{\sqrt{N}}Y_N(t/\sqrt{N}) +  \frac{1}{\sqrt{N}} \frac{2\b}{1+\b}Z_N(t/\sqrt{N}) \right)
\]
and
\[
\l_N^{-}(t) = f \left( \frac{1}{\sqrt{N}}Y_N(t/\sqrt{N}) -  \frac{1}{\sqrt{N}} \frac{2}{1+\b}Z_N(t/\sqrt{N}) \right)
\]
We recall (see \eqref{def:X}) that $m_N^{\pm}(t) \geq 0$, and therefore $Y_N(t) \geq 0$ (recall $\gamma \leq 1$). Thus
\[
|Z_N(t)| \leq \frac{(1-\g)\sqrt{N}}{2} \left[  m_N^+(\sqrt{N}t) +  m_N^-(\sqrt{N}t)\right]
\]
which implies the bound
\be{boundZ}
|Z_N(t)| \leq C(\b,\g) Y_N(t).
\ee
holds for some constant $C(\b,\g)$ depending on $\b$ and $\g$.
We are going to see, through a sequence of lemmas, that the process $Y_N$ has a nontrivial limit as $N \ra +\infty$, while $Z_N$ converges to zero in a suitable sense.
\begin{lemma} \label{lemma:Z1}
The following identity holds for all $t \geq 0$:
\be{eqZ_N}
\begin{split}
Z_N(t) & = E_N + F_N t  - \a \sqrt{N} \int_0^{t} Z_N(s) ds \\ & ~~~~~~+ \frac{1-\g}{2 \sqrt{N}} \int_{[0,\sqrt{N} t]} {\bf 1}_{[0,\l_N^+(s))}(u) \sum_{i=1}^N n_i^+(ds,du) - \frac{1-\g}{2 \sqrt{N}} \int_{[0,\sqrt{N} t]} {\bf 1}_{[0,\l_N^-(s))}(u) \sum_{i=1}^N n_i^-(ds,du),
\end{split}
\ee
where 
\[
E_N = \frac{1-\g}{2} \sqrt{N} [a_N^+ - a_N^-] \ \ \ F_N = \frac{1-\g}{2} \sqrt{N} [b_N^+ - b_N^-].
\]
\end{lemma}
\begin{proof}
This follows immediately from Lemma \ref{lemma:mN} and the definition of $Z_N$.
\end{proof}

Now a truncation argument is needed for later use. 
For a fixed but arbitrary $h>0$ define 
\[
\tY_N(t) := Y_N(t) \wedge h.
\]
Then, we define $\tZ_N$ to be the process such that \eqref{eqZ_N} holds for all $t \geq 0$ after replacing $Y_N$ and $Z_N$ by $\tY_N$ and $\tZ_N$ (these processes appear as arguments of $\l^{\pm}$). 
More precisely, given a realization of $n_i^{\pm}$, $i=1,\ldots,N$, which determines the trajectory of $Y_N$ and $\tY_N$, define $\tZ_N$ as the process that satisfies the equation
\be{eqtZ_N}
\begin{split}
\tZ_N(t) & = E_N + F_N t  - \a \sqrt{N} \int_0^{t} \tZ_N(s) ds \\ & ~~~~~~+ \frac{1-\g}{2 \sqrt{N}} \int_{[0,\sqrt{N} t]} {\bf 1}_{[0,\tilde{\l}_N^+(s))}(u) \sum_{i=1}^N n_i^+(ds,du) - \frac{1-\g}{2 \sqrt{N}} \int_{[0,\sqrt{N} t]} {\bf 1}_{[0,\tilde{\l}_N^-(s))}(u) \sum_{i=1}^N n_i^-(ds,du),
\end{split}
\ee
where 
\be{tl+}
\tilde{\l}_N^{+}(t) = f \left( \frac{1}{\sqrt{N}}\tY_N(t/\sqrt{N}) +  \frac{1}{\sqrt{N}} \frac{2\b}{1+\b}\tZ_N(t/\sqrt{N}) \right)
\ee
and
\be{tl-}
\tilde{\l}_N^{-}(t) = f \left( \frac{1}{\sqrt{N}}\tY_N(t/\sqrt{N}) -  \frac{1}{\sqrt{N}} \frac{2}{1+\b}\tZ_N(t/\sqrt{N}) \right)
\ee
Similarly to what is done in Remark \ref{rem:sde}, the trajectory of $\tZ_N$ can be reconstructed pathwise, solving the deterministic part of the equation between jumps, and performing the jumps induced by the points of the $n_i^{\pm}$'s and the moving intervals $[0,\tilde{\l}_N^{\pm}(s))$. Since the truncation may not guarantee the fact that the arguments of $f$ in \eqref{tl+} and \eqref{tl-} are positive, we extend $f$ by setting $f(x) = 0$ for $x<0$. The relevant outcome of this truncation is the boundedness of $\tY_N$, and the fact that $\tY_N(t) = Y_N(t)$ and $\tZ_N(t) = Z_N(t)$ up to the stopping time 
\[
\tau^N_h := \inf\{t>0 : Y_N(t) > h\}.
\]

It what follows $h>0$ is fixed but arbitrary, so all statement are meant to hold for all $h>0$. In the rest of the paper we will use extensively the compensated Poisson Random Measures
\be{comppoint}
\tilde{n}_i^{\pm}(dt,du)  := n_i^{\pm}(dt,du) - dt du.
\ee
\begin{lemma} \label{lemma:Z2}
The following bound holds for every $t>0$:
\be{collZ}
\sup_N N^{\frac12} \E\left[\tZ_N^4(t)\right] < +\infty.
\ee
\end{lemma}
\begin{proof}
It follows from \eqref{eqtZ_N} that  (see e.g. Theorem 31, Ch. II in \cite{Protter2005})
\be{tZ2}
\begin{split}
\tZ_N^4(t) & = E_N^4 + 4 F_N \int_0^t \tZ^2_N(s)ds - 4 \a \sqrt{N} \int_0^t \tZ_N^4(s)ds \\ & + \int_{[0,\sqrt{N} t]} {\bf 1}_{[0,\tilde{\l}_N^+(s))}(u) \left[\left(\tZ_N\left(\frac{s-}{\sqrt{N}}\right) + \frac{1-\g}{2 \sqrt{N}} \right)^4 - \tZ^4_N\left(\frac{s-}{\sqrt{N}}\right) \right]\sum_{i=1}^N n_i^+(ds,du) \\ & + \int_{[0,\sqrt{N} t]} {\bf 1}_{[0,\tilde{\l}_N^-(s))}(u) \left[\left(\tZ_N\left(\frac{s-}{\sqrt{N}}\right) - \frac{1-\g}{2 \sqrt{N}} \right)^4 - \tZ^4_N\left(\frac{s-}{\sqrt{N}}\right) \right] \sum_{i=1}^N n_i^-(ds,du)
\end{split}
\ee
Using then \eqref{comppoint} we obtain
\be{tZ2mart}
\begin{split}
\tZ_N^4(t) & = E_N^4 + 4 F_N \int_0^t \tZ^2_N(s)ds - 4 \a \sqrt{N} \int_0^t \tZ_N^4(s)ds \\ & + N \int_0^{\sqrt{N}t}\left[\frac{(1-\g)^4}{16N} + \frac{(1-\g)^3}{ 2 N^{\frac32}}\tZ_N\left(\frac{s}{\sqrt{N}}\right) + 3\frac{(1-\g)^2}{ 2 N}\tZ^2_N\left(\frac{s}{\sqrt{N}}\right) + 2\frac{(1-\g)}{ \sqrt{N}}\tZ^3_N\left(\frac{s}{\sqrt{N}}\right)\right] \\ & ~~~~~~~~~~~~~  f \left( \frac{1}{\sqrt{N}}\tY_N\left(\frac{s}{\sqrt{N}}\right) +  \frac{1}{\sqrt{N}}\frac{2\b}{1+\b}\tZ_N\left(\frac{s}{\sqrt{N}}\right) \right) ds \\ & + N \int_0^{\sqrt{N}t}\left[\frac{(1-\g)^4}{16N} - \frac{(1-\g)^3}{ 2 N^{\frac32}}\tZ_N\left(\frac{s}{\sqrt{N}}\right) + 3\frac{(1-\g)^2}{ 2 N}\tZ^2_N\left(\frac{s}{\sqrt{N}}\right) - 2\frac{(1-\g)}{ \sqrt{N}}\tZ^3_N\left(\frac{s}{\sqrt{N}}\right)\right] \\ & ~~~~~~~~~~~~~  f \left( \frac{1}{\sqrt{N}}\tY_N\left(\frac{s}{\sqrt{N}}\right) -  \frac{1}{\sqrt{N}}\frac{2\b}{1+\b}\tZ_N\left(\frac{s}{\sqrt{N}}\right) \right) ds  \\ & + M_N^{\tZ^4}(t) \\ & = E_N^4 + 4 F_N \int_0^t \tZ^2_N(s)ds - 4 \a \sqrt{N} \int_0^t \tZ_N^4(s)ds \\ & + N^{\frac32} \int_0^{t}\left[\frac{(1-\g)^4}{16N} + \frac{(1-\g)^3}{ 2 N^{\frac32}}\tZ_N\left(s\right) + 3\frac{(1-\g)^2}{ 2 N}\tZ^2_N\left(s\right) + 2\frac{(1-\g)}{ \sqrt{N}}\tZ^3_N\left(s\right)\right] \\ & ~~~~~~~~~~~~~  f \left( \frac{1}{\sqrt{N}}\tY_N\left(s\right) +  \frac{1}{\sqrt{N}}\frac{2\b}{1+\b}\tZ_N\left(s\right) \right) ds \\ & + N^{\frac32} \int_0^{t}\left[\frac{(1-\g)^4}{16N} - \frac{(1-\g)^3}{ 2 N^{\frac32}}\tZ_N\left(s\right) + 3\frac{(1-\g)^2}{ 2 N}\tZ^2_N\left(s\right) - 2\frac{(1-\g)}{ \sqrt{N}}\tZ^3_N\left(s\right)\right] \\ & ~~~~~~~~~~~~~  f \left( \frac{1}{\sqrt{N}}\tY_N\left(s\right) -  \frac{1}{\sqrt{N}}\frac{2\b}{1+\b}\tZ_N\left(s\right) \right) ds \\ & + M_N^{\tZ^4}(t),
\end{split}
\ee
where 
\[
\begin{split}
M_N^{\tZ^4} = &  \int_{[0,\sqrt{N} t]} {\bf 1}_{[0,\tilde{\l}_N^+(s))}(u) \left[\left(\tZ_N\left(\frac{s-}{\sqrt{N}}\right) + \frac{1-\g}{2 \sqrt{N}} \right)^4 - \tZ^4_N\left(\frac{s-}{\sqrt{N}}\right) \right]\sum_{i=1}^N \tilde{n}_i^+(ds,du) \\ & + \int_{[0,\sqrt{N} t]} {\bf 1}_{[0,\tilde{\l}_N^-(s))}(u) \left[\left(\tZ_N\left(\frac{s-}{\sqrt{N}}\right) - \frac{1-\g}{2 \sqrt{N}} \right)^4 - \tZ^4_N\left(\frac{s-}{\sqrt{N}}\right) \right] \sum_{i=1}^N \tilde{n}_i^-(ds,du)
\end{split}
\]
is a mean zero martingale. Taking the expectation in \eqref{tZ2mart} we obtain
\be{tZ2der}
\begin{split}
\frac{d}{dt} \E\left[\tZ_N^4(t)\right] &  =  4 F_N \E\left[\tZ_N(t)\right] -4\a \sqrt{N} \E\left[\tZ_N^4(t)\right] \\ & + N^{3/2}\E\left[\left[\frac{(1-\g)^4}{16N^2} + \frac{3(1-\g)^2}{ 2N}\tZ^2_N\left(t\right)\right] \left[ f \left( \frac{1}{\sqrt{N}}\tY_N\left(t\right) +  \frac{1}{\sqrt{N}}\frac{2\b}{1+\b}\tZ_N\left(s\right) \right)  \right. \right. \\ & ~~~~~~~~~~~~~~~~~~~~~~~~~~~~~~~~~~~~~~~~~~~~~~~~~~~~~~\left. \left. + f \left( \frac{1}{\sqrt{N}}\tY_N\left(t\right) -  \frac{1}{\sqrt{N}}\frac{2}{1+\b}\tZ_N\left(s\right) \right)\right] \right]  \\ & + N^{3/2}\E\left[\left[\frac{(1-\g)^3}{2N^{\frac32}}\tZ_N + \frac{2(1-\g)}{ \sqrt{N}}\tZ^3_N\left(t\right)\right] \left[ f \left( \frac{1}{\sqrt{N}}\tY_N\left(t\right) +  \frac{1}{\sqrt{N}}\frac{2\b}{1+\b}\tZ_N\left(s\right) \right)  \right. \right. \\ & ~~~~~~~~~~~~~~~~~~~~~~~~~~~~~~~~~~~~~~~~~~~~~~~~~~~~~~\left. \left. - f \left( \frac{1}{\sqrt{N}}\tY_N\left(t\right) -  \frac{1}{\sqrt{N}}\frac{2}{1+\b}\tZ_N\left(s\right) \right)\right] \right] 
\end{split}
\ee
Using the fact that $f$ is concave for $x>0$ we have that $f(x) \leq f'(0)|x|$ for all $x \in \R$. So
\begin{multline} \label{conc}
 f \left( \frac{1}{\sqrt{N}}\tY_N\left(t\right) +  \frac{1}{\sqrt{N}}\frac{2\b}{1+\b}\tZ_N\left(t\right) \right) +  f \left( \frac{1}{\sqrt{N}}\tY_N\left(t\right) -  \frac{1}{\sqrt{N}}\frac{2}{1+\b}\tZ_N\left(t\right) \right)
 \\ \leq \frac{2}{\sqrt{N}}f'(0)\left[ \tY_N(t) + 2 |\tZ_N(t)| \right].
 \end{multline}
Moreover $f$ is Lipschitz with constant $f'(0)$, so
\begin{equation} \label{Lip}
 \left| f \left( \frac{1}{\sqrt{N}}\tY_N\left(t\right) +  \frac{1}{\sqrt{N}}\frac{2\b}{1+\b}\tZ_N\left(t\right) \right) -  f \left( \frac{1}{\sqrt{N}}\tY_N\left(t\right) -  \frac{1}{\sqrt{N}}\frac{2}{1+\b}\tZ_N\left(t\right) \right) \right| 
 \leq \frac{2 f'(0)}{\sqrt{N}} |\tZ_N(t)|.
\end{equation}
Summing all up we obtain
\be{tZ2derest}
\begin{split}
\frac{d}{dt} \E\left[\tZ_N^4(t)\right]  \leq  & - 4[\a-(1-\g) f'(0)] \sqrt{N} \E\left[\tZ_N^4(t)\right] \\ &
+ [4F_N + 6(1-\g)^2 f'(0)]\E\left[|\tZ_N(t)|^3\right] + 3 (1-\g)^2 f'(0) \E\left[ \tZ_N^2(t) \tY_N(y) \right] \\ & + (1-\g)^3 \frac{f'(0)}{\sqrt{N}} \E\left[\tZ_N^2(t)\right] + \frac{(1-\g)^4}{8N} f'(0)\E\left[2|\tZ_N| +  \tY_N(t)\right].
\end{split}
\ee
Using the fact that the sequence $F_N$ is bounded (see \eqref{assab}), that $\tY_N(t) \leq h$  and that $|\tZ_N|^i \leq 1 + |\tZ_N|^3$ for $i=1,2$, we have that 
\be{tZ2derest2}
\begin{split}
\frac{d}{dt} \E\left[\tZ_N^4(t)\right]    & \leq   A + B\E\left[ |\tZ_N|^3 \right] - 4[\a-(1-\g) f'(0)] \sqrt{N} \E\left[\tZ_N^4(t)\right]  \\ & \leq A + B\left(\E\left[ |\tZ_N|^4 \right]\right)^{\frac34} - 4[\a-(1-\g) f'(0)] \sqrt{N} \E\left[\tZ_N^4(t)\right],
\end{split}
\ee
for suitable constants $A,B>0$, where in the last step we have used Jensen's inequality. 
Thus $x_N(t) := \E\left[\tZ_N^4(t)\right] $ 
satisfies the differential inequality:
\be{}
\frac{d}{dt}x_N(t)  \leq A + B x^{\frac34}_N(t) - 2[\a-(1-\g) f'(0)] \sqrt{N} x_N(t).
\ee
Considering also that $x_N(0) = E_N^4$ which is bounded by a constant $C$ (see \eqref{assab}), we have, by the comparison theorem for differential equations, that $x_N(t) \leq y_N(t)$, where 
\[
\begin{split}
\frac{d}{dt}y_N(t) & = A + B y^{\frac34}_N(t) - 2[\a-(1-\g) f'(0)] \sqrt{N} y_N(t) \\
y_N(0) & = C.
\end{split}
\]
From \eqref{intro:critical}, observing that $\a-(1-\g) f'(0) = (1+\b\g) f'(0) -(1-\g) f'(0) = \b(1+\g)f'(0) >0$, the conclusion now follows from Lemma \ref{lemma:collapsing}.
\end{proof}

From the pointwise estimate \eqref{collZ} we can derive a uniform estimate by using Theorem \ref{collapsing}. This estimate is not needed for the convergence of the volatility process, but will be used for proving the convergence of the price.
\begin{lemma} \label{lemma:Z3}
For all $\e>0$,
\be{collZpr}
\sup_{t \in [\e,T]}  \tZ_N^2(t) \ \longrightarrow \ 0 \ \ \mbox{in probability and in $L^1$, as $N \rightarrow +\infty$}
\ee
\end{lemma}
\begin{proof}
First we identify the index $n$ appearing in Theorem \ref{collapsing} with $N^{\frac{3}{16}}$, so all processes are indexed by $N$ rather than $n$. Set $X_N(t) = \tZ_N^2(t)$. By \eqref{eqtZ_N} following the argument that led to \eqref{tZ2}, we get
\be{tZ2bis}
\begin{split}
\tZ_N^2(t) & = E_N^2 + 2 F_N \int_0^t \tZ_N(s)ds - 2 \a \sqrt{N} \int_0^t \tZ_N^2(s)ds \\ & + \int_{[0,\sqrt{N} t]} {\bf 1}_{[0,\tilde{\l}_N^+(s))}(u) \left[\left(\tZ_N\left(\frac{s-}{\sqrt{N}}\right) + \frac{1-\g}{2 \sqrt{N}} \right)^2 - \tZ^2_N\left(\frac{s-}{\sqrt{N}}\right) \right]\sum_{i=1}^N n_i^+(ds,du) \\ & + \int_{[0,\sqrt{N} t]} {\bf 1}_{[0,\tilde{\l}_N^-(s))}(u) \left[\left(\tZ_N\left(\frac{s-}{\sqrt{N}}\right) - \frac{1-\g}{2 \sqrt{N}} \right)^2 - \tZ^2_N\left(\frac{s-}{\sqrt{N}}\right) \right] \sum_{i=1}^N n_i^-(ds,du) \\
&= E_N^2 + 2 F_N \int_0^t \tZ_N(s)ds - 2 \a \sqrt{N} \int_0^t \tZ_N^2(s)ds \\ & + \int_{[0,t]} {\bf 1}_{[0,\tilde{\l}_N^+(\sqrt{N}s))}(u) \left[\left(\tZ_N\left(s-\right) + \frac{1-\g}{2 \sqrt{N}} \right)^2 - \tZ^2_N\left(s-\right) \right]\sum_{i=1}^N n_i^+(\sqrt{N}ds,du) \\ & + \int_{[0, t]} {\bf 1}_{[0,\tilde{\l}_N^-(\sqrt{N}s))}(u) \left[\left(\tZ_N\left(s-\right) - \frac{1-\g}{2 \sqrt{N}} \right)^2 - \tZ^2_N\left(s-\right) \right] \sum_{i=1}^N n_i^-(\sqrt{N}ds,du)
\end{split}
\ee
This can be rewritten in the form
\begin{equation} \label{Xcoll} dX_N(t) = S_N(t)dt + \int_{[0,t]\times \mathcal{Y}} f_N(s-, y) [\Lambda_N(ds,dy) - A_N(s,dy)ds]\end{equation}
with 
\[
\mathcal{Y} = [0,+\infty) \times \{+,-\}, \ \ \Lambda_N(ds,du, \pm) = 
\sum_{i=1}^N n^{\pm}_i(\sqrt{N}ds,du), \ \ A_N(s,du,\pm) = N^{\frac32} du
\]
\begin{equation} \label{effeN}
f_N(s,u,\pm) =  {\bf 1}_{[0, \tilde{\lambda}^{\pm}(\sqrt{N} s))}(u) \left[ \left(\tZ_N(s) \pm \frac{1-\g}{2\sqrt{N}} \right)^2 - \tZ_N^2(s) \right],
\end{equation}
\begin{equation} \label{esseN}
\begin{split}
S_N(t) = 2 F_N \tZ_N(t) - 2 \a \sqrt{N} \tZ_N^2(t) & + N^{\frac32}\tilde{\l}_N^{+}(\sqrt{N}t)\left[\left(\tZ_N\left(s-\right) + \frac{1-\g}{2 \sqrt{N}} \right)^2 - \tZ^2_N\left(s-\right) \right] \\ & + N^{\frac32}\tilde{\l}_N^{-}(\sqrt{N}t)\left[\left(\tZ_N\left(s-\right) - \frac{1-\g}{2 \sqrt{N}} \right)^2 - \tZ^2_N\left(s-\right) \right].
\end{split}
\end{equation}
We now verify that conditions in Theorem \ref{collapsing} hold, with the only difference that the initial time is $\e>0$ rather than zero.  This is due to the fact that we will use \eqref{collZ}, valid for $t>0$ only, to control the initial conditions (see \eqref{condc2}). We set $d=2$, $\a_N = N^{\frac18}$, $\beta_N = 1$, and $\tau_N = \tau_h^N$. With these choices requirements \eqref{condc1} are satisfied. Now, using Lemma \ref{lemma:Z2} we have
\[
\mathbb{E}\left[X_N^d(\e)\right] = \mathbb{E}\left[\tZ_N^4(\e)\right] \leq  A N^{-\frac12} \leq B N^{-\frac14} = B \a_N^{-d}
\]
for some constants $A$ and $B$, so condition \eqref{condc2} is checked. 
For condition \eqref{condc3}, we make estimates similar to those in Lemma \ref{lemma:Z2}.  Recalling that
\[
\tilde{\l}_N^{+}(\sqrt{N}t) = f \left( \frac{1}{\sqrt{N}}\tY_N(t) +  \frac{1}{\sqrt{N}} \frac{2\b}{1+\b}\tZ_N(t) \right)
\]
and
\[
\tilde{\l}_N^{-}(\sqrt{N}t) = f \left( \frac{1}{\sqrt{N}}\tY_N(t) -  \frac{1}{\sqrt{N}} \frac{2}{1+\b}\tZ_N(t) \right),
\]
we have
\begin{multline*}
N^{\frac32}\tilde{\l}_N^{+}(\sqrt{N}t)\left[\left(\tZ_N\left(s-\right) + \frac{1-\g}{2 \sqrt{N}} \right)^2 - \tZ^2_N\left(s-\right) \right]  + N^{\frac32}\tilde{\l}_N^{-}(\sqrt{N}t)\left[\left(\tZ_N\left(s-\right) - \frac{1-\g}{2 \sqrt{N}} \right)^2 - \tZ^2_N\left(s-\right) \right] \\ = N^{\frac32} \frac{(1-\g)^2}{4N} \left[f \left( \frac{1}{\sqrt{N}}\tY_N(t) +  \frac{1}{\sqrt{N}} \frac{2\b}{1+\b}\tZ_N(t) \right) + f \left( \frac{1}{\sqrt{N}}\tY_N(t) -  \frac{1}{\sqrt{N}} \frac{2}{1+\b}\tZ_N(t) \right) \right] \\
+ N^{\frac32} \frac{1-\g}{\sqrt{N}} \tZ_N(t) \left[f \left( \frac{1}{\sqrt{N}}\tY_N(t) +  \frac{1}{\sqrt{N}} \frac{2\b}{1+\b}\tZ_N(t) \right) - f \left( \frac{1}{\sqrt{N}}\tY_N(t) -  \frac{1}{\sqrt{N}} \frac{2}{1+\b}\tZ_N(t) \right) \right] \\ \leq \frac{(1-\g)^2}{2} f'(0)\left[ \tY_N(t) + 2 |\tZ_N(t)| \right] + 2f'(0)(1-\g) \sqrt{N} \tZ^2_N(t),
\end{multline*}
where we have used \eqref{conc} and \eqref{Lip}. Summing up
\begin{equation} \label{estS_N}
S_N(t) \leq 2 F_N \tZ_N(t) + \frac{(1-\g)^2}{2} f'(0)\left[ \tY_N(t)+ 2 |\tZ_N(t)| \right] - 2 (\a-f'(0)(1-\g))  \sqrt{N} \tZ_N^2(t).
\end{equation}
We observe that for $t < \tau_h^N$ we have $\tY_N(y) = Y_N(t) \leq h$ and, by \eqref{boundZ},
\be{esttZ}
|\tZ_N(t)| = |Z_N(t)| \leq C(\b,\g) Y_N(t) \leq C(\b,\g)h.
\ee
So condition \eqref{condc3} follows from \eqref{estS_N},  and the facts that $\a-f'(0)(1-\g)>0$ and $\sqrt{N} \geq n = N^{\frac{3}{16}}$. We now consider conditions \eqref{condc4} and \eqref{condc5}. Note that, by \eqref{effeN}, 
\[
|f_N(s,u,\pm)| \leq \frac{A}{\sqrt{N}} \leq B \a_N^{-1}
\]
for some constants $A$ and $B$; moreover
\[
\begin{split}
\int_{\mathcal{Y}} (f_N(t,y))^2 A_N(t,dy)  = &  N^{\frac32}\tilde{\lambda}^+ (\sqrt{N} t) \left[ \left(\tZ_N(s) + \frac{1-\g}{2\sqrt{N}} \right)^2 - \tZ_N^2(s) \right]^2 \\
& + N^{\frac32}\tilde{\lambda}^- (\sqrt{N} t) \left[ \left(\tZ_N(s) - \frac{1-\g}{2\sqrt{N}} \right)^2 - \tZ_N^2(s) \right]^2  \leq A
\end{split}
\]
for some constant $A$, as, for some $C>0$, 
\[
\tilde{\lambda}^{\pm} (\sqrt{N} t) \leq \frac{C}{\sqrt{N}}
\]
by definition of $\tilde{\lambda}^{\pm}$ and boundedness of $\tZ_N$ and $\tY_N$ up to time $\tau_N^h$, and
\[
\left[ \left(\tZ_N(s) \pm \frac{1-\g}{2\sqrt{N}} \right)^2 - \tZ_N^2(s) \right]^2 \leq \frac{C}{N}.
\]
Thus all conditions in Theorem \ref{collapsing} are verified, and the conclusion follows.

\end{proof}

Now we are ready to prove the part of Theorem \ref{th:main} which concerns the volatility process. More explicitly, we prove the following result.

\begin{proposition} \label{th:volatility}
Suppose the criticality condition 
\be{critical}
\a = f'(0)(1+\b\g)
\ee
holds. Then the sequence $(Y_N(t))_{t \in [0,T]}$ converges in distribution to the unique solution of the SDE
\be{limvol}
\begin{split}
dy(t)  & = \left[\frac{1+\b\g}{1+\b} b^+ + \b\frac{1+\b\g}{1+\b} b^- + \frac12 (1+\b\g) f''(0) y^2(t) \right] dt  \\ & ~~~~~ + \frac{\sqrt{f'(0) (1+\b^2)}(1+\b\g)}{1+\b} \sqrt{y(t)} dB(t) \\
y(0) & = \frac{1+\b\g}{1+\b} a^+ + \b\frac{1+\b\g}{1+\b} a^-,
\end{split}
\ee
where $B$ is a simple Brownian motion.
\end{proposition}
\begin{proof} Recall $M^\pm_N$ in Lemma \ref{lemma:mN}
Using \eqref{macsde} we obtain
\be{Y}
\begin{split}
Y_N(t) & =  \sqrt{N} \left[\frac{1+\b\g}{1+\b} m_N^+(\sqrt{N}t) + \b \frac{1+\b\g}{1+\b} m_N^-(\sqrt{N}t)\right] \\
& =  \sqrt{N} \left[ - \a \int_0^{\sqrt{N} t} \frac{1+\b\g}{1+\b}m_N^+(s) + \b \frac{1+\b\g}{1+\b} m_N^-(s)  ds  \right]\\
& ~~~ + \sqrt{N} \left[\frac{1+\b\g}{1+\b} \int_0^{\sqrt{N} t} f\left(\frac{1}{\sqrt{N}}Y_N\left(s/\sqrt{N}\right) +  \frac{1}{\sqrt{N}}\frac{2\b}{1+\b}Z_N\left(s/\sqrt{N}\right) \right) ds \right. \\
& ~~~ + \left. \b \frac{1+\b\g}{1+\b}  \int_0^{\sqrt{N} t} f\left( \frac{1}{\sqrt{N}}Y_N(s/\sqrt{N}) -  \frac{1}{\sqrt{N}} \frac{2}{1+\b}Z_N(s/\sqrt{N}) \right) ds \right] \\
& ~~~ +  \sqrt{N} \left[ \frac{1+\b\g}{1+\b} \left(a_N^+ + b_N^+ t \right) + \b \frac{1+\b\g}{1+\b} \left(a_N^- + b_N^- t \right) \right] \\
& ~~~ +  \sqrt{N} \left[  \frac{1+\b\g}{1+\b}M_N^+(\sqrt{N} t) + \b \frac{1+\b\g}{1+\b} M_N^-(\sqrt{N} t) \right] \\
& = C_N + D_N t  -\a \sqrt{N} \int_0^t Y_N(s)  ds  \\ & ~~~+ N \left[ \frac{1+\b\g}{1+\b} \int_0^{ t} f\left(\frac{1}{\sqrt{N}}Y_N\left(s\right) +  \frac{1}{\sqrt{N}}\frac{2\b}{1+\b}Z_N\left(s\right)  \right) ds \right. \\
& ~~~ + \left. \b\frac{1+\b\g}{1+\b}  \int_0^{ t} f\left( \frac{1}{\sqrt{N}}Y_N(s) -  \frac{1}{\sqrt{N}} \frac{2}{1+\b}Z_N(s)  \right) ds \right]  \\ & ~~~ + M_N^Y(t)
\end{split}
\ee
where 
\[
\begin{split}
C_N & =   \sqrt{N} \left[ \frac{1+\b\g}{1+\b} a_N^+ + \b \frac{1+\b\g}{1+\b} a_N^- \right]  \\ D_N & =  \sqrt{N} \left[ \frac{1+\b\g}{1+\b} b_N^+ + \b \frac{1+\b\g}{1+\b} b_N^- \right],
\end{split}
\]

\[
M_N^{\pm}(t) := \int_{[0,t]\times [0,+\infty)} {\bf 1}_{[0,\l^{\pm}(s-))}(u) \frac{1}{N} \sum_{i=1}^N \tilde{n}_i^{\pm}(ds,du)
\]
and
\[
M_N^Y(t) =  \sqrt{N} \left[  \frac{1+\b\g}{1+\b}M_N^+(\sqrt{N} t) + \b \frac{1+\b\g}{1+\b} M_N^-(\sqrt{N} t) \right]
\]
 is a martingale with conditional quadratic variation
\begin{multline} \label{quadvarMY}
\sqrt{N} \left[ \left(\frac{1+\b\g}{1+\b}\right)^2\int_0^t  f\left(\frac{1}{\sqrt{N}}Y_N\left(s\right) +  \frac{1}{\sqrt{N}}\frac{2\b}{1+\b}Z_N\left(s\right) \right) ds \right. \\  
\left.+ \left(\b\frac{1+\b\g}{1+\b}\right) ^2\int_0^t  f\left( \frac{1}{\sqrt{N}}Y_N(s) -  \frac{1}{\sqrt{N}} \frac{2}{1+\b}Z_N(s)  \right) ds \right].
\end{multline}
We now use the Taylor expansion
\be{taylor}
f(u) = f'(0) u + \frac12 f''(0) u^2 + R(u) u^3,
\ee
where $R(u)$ is a  bounded function (the boundedness of $R(u)$ follows from the boundedness of the third derivative). After a simple computation, we rewrite \eqref{Y} as
\be{Y2}
\begin{split}
Y_N(t) & = C_N + D_N t + \sqrt{N} \int_0^t [-\a + f'(0)(1+\b \g)] Y_N(s) ds  \\
& ~~~ + \frac12(1+\b\g) f''(0) \int_0^t Y_N^2(s) ds + A(\b,\g) f''(0) \int_0^t Y_N(s)Z_N(s) ds + B(\b,\g) f''(0) \int_0^t Z_N^2(s) ds  \\
& ~~~ + \frac{1}{\sqrt{N}} C(\b,\g) \int_0^t R\left(\frac{1}{\sqrt{N}}Y_N\left(s\right) +  \frac{1}{\sqrt{N}}\frac{2\b}{1+\b}Z_N\left(s\right)  \right) \left(\frac{1}{\sqrt{N}}Y_N\left(s\right) +  \frac{1}{\sqrt{N}}\frac{2\b}{1+\b}Z_N\left(s\right)  \right)^3 ds \\ & ~~~ + \frac{1}{\sqrt{N}}D(\b,\g)  \int_0^t R\left( \frac{1}{\sqrt{N}}Y_N(s) -  \frac{1}{\sqrt{N}} \frac{2}{1+\b}Z_N(s) \right) \left( \frac{1}{\sqrt{N}}Y_N(s) -  \frac{1}{\sqrt{N}} \frac{2}{1+\b}Z_N(s) \right)^3 ds \\ & ~~~ +  M_N^Y(t)
\end{split}
\ee
where $A(\b,\g)$, $B(\b,\g)$, $C(\b,\g)$ and $D(\b,\g)$ are constants depending on $\b,\g$, whose precise value will not be relevant.
It should be remarked that in this expansion the term of order $\sqrt{N} \int_0^T Z_N(s) ds$ has zero coefficient: this is the main motivation for the specific choice of the combinators in the definition of $Y_N$ and $Z_N$.

We are now ready to 
use Theorem \ref{thm:diffusion:approx}, for the one-dimensional process $Y_N(t)$. We choose 
\be{B}
\begin{split}
B_N(t) &  := C_N + D_N t \\ & ~~~ + \frac12(1+\b\g) f''(0) \int_0^t Y_N^2(s) ds + A(\b,\g) f''(0) \int_0^t Y_N(s)Z_N(s) ds + B(\b,\g) f''(0) \int_0^t Z_N^2(s) ds  \\
& ~~~ + \frac{1}{\sqrt{N}} C(\b,\g) \int_0^t R\left(\frac{1}{\sqrt{N}}Y_N\left(s\right) +  \frac{1}{\sqrt{N}}\frac{2\b}{1+\b}Z_N\left(s\right)  \right) \left(\frac{1}{\sqrt{N}}Y_N\left(s\right) +  \frac{1}{\sqrt{N}}\frac{2\b}{1+\b}Z_N\left(s\right)  \right)^3 ds \\ & ~~~ + \frac{1}{\sqrt{N}}D(\b,\g)  \int_0^t R\left( \frac{1}{\sqrt{N}}Y_N(s) -  \frac{1}{\sqrt{N}} \frac{2}{1+\b}Z_N(s) \right) \left( \frac{1}{\sqrt{N}}Y_N(s) -  \frac{1}{\sqrt{N}} \frac{2}{1+\b}Z_N(s) \right)^3 ds,
\end{split}
\ee
\be{b}
b(x) := \left[2\frac{1+\b\g}{1+\b} b^+ + 2\b \frac{1+\b\g}{1+\b} b^- + \frac12 (1+\b\g) f''(0) x^2 \right],
\ee
\be{A}
\begin{split}
A_N(t) & := \frac{ \sqrt{N} }{4}\left[ \left(\frac{2}{1+\b}+\frac{2\b}{1+\b}\g\right)^2\int_0^t  f\left(\frac{1}{\sqrt{N}}Y_N\left(s\right) +  \frac{1}{\sqrt{N}}\frac{2\b}{1+\b}Z_N\left(s\right) \right) ds \right. \\ & ~~~~~~~ 
\left.+ \left(\frac{2\b}{1+\b}(1-\g) + 2\b\g\right) ^2\int_0^t  f\left( \frac{1}{\sqrt{N}}Y_N(s) -  \frac{1}{\sqrt{N}} \frac{2}{1+\b}Z_N(s)  \right) ds \right],
\end{split}
\ee
\be{a}
a(x) :=  \frac{\a (1+\b^2)(1+\b\g)}{(1+\b)^2}|x|.
\ee
By   \eqref{Y2} it follows that 
\[
Y_N - B_N = M_N^Y
\]
 and therefore is a martingale that, by \eqref{quadvarMY}, has  conditional quadratic variation $A_N$, i.e. $(M_N^Y)^2 - A_N$ is a martingale, as required in Theorem \ref{thm:diffusion:approx}. 

 Conditions \eqref{eq:hyp4.4}, \eqref{eq:hyp4.5} and \eqref{eq:hyp4.3} are trivial, since $A^N$ and $B^N$ are continuous in $t$ and the jumps of $Y_N$ have order $\frac{1}{\sqrt{N}}$:  indeed, by \eqref{macsde}, the jumps of $m_N^{\pm}$ are of size $\frac{1}{N}$, since simultaneous jumps are not allowed for independent Poisson processes. By Theorem \ref{thm:eq.quad} the SDE \eqref{limvol} has a unique solution, and 
\[
\lim_{N \ra +\infty} Y_N(0) = \lim_{N \ra +\infty}  C_N  = y(0).
\]
Thus, to complete the proof, we need to establish \eqref{eq:hyp4.7} and \eqref{eq:hyp4.6}. We begin by \eqref{eq:hyp4.6}. Note that 
\begin{multline} \label{estB}
B_N(t) - \int_0^t b(Y_N(s)) ds   = \left[ D_N - 2\frac{1+\b\g}{1+\b} b^+ - 2\b \frac{1+\b\g}{1+\b} b^- \right]t \\ +
A(\b,\g) f''(0) \int_0^t Y_N(s)Z_N(s) ds + B(\b,\g) f''(0) \int_0^t Z_N^2(s) ds \\
 + \frac{1}{\sqrt{N}} C(\b,\g) \int_0^t R\left(\frac{1}{\sqrt{N}}Y_N\left(s\right) +  \frac{1}{\sqrt{N}}\frac{2\b}{1+\b}Z_N\left(s\right)  \right) \left(\frac{1}{\sqrt{N}}Y_N\left(s\right) +  \frac{1}{\sqrt{N}}\frac{2\b}{1+\b}Z_N\left(s\right)  \right)^3 ds \\ + \frac{1}{\sqrt{N}}D(\b,\g)  \int_0^t R\left( \frac{1}{\sqrt{N}}Y_N(s) -  \frac{1}{\sqrt{N}} \frac{2}{1+\b}Z_N(s) \right) \left( \frac{1}{\sqrt{N}}Y_N(s) -  \frac{1}{\sqrt{N}} \frac{2}{1+\b}Z_N(s) \right)^3 ds
\end{multline}
To prove \eqref{eq:hyp4.6} we show that all terms in the r.h.s of \eqref{estB} converge to zero in probability after replacing $t$ by $t \wedge  \tau^N_h$, with  $\tau^N_h := \inf\{t>0 : Y_N(t) > h\}$. The first term is easy, since 
\[
\lim_{N \ra +\infty} D_N  = 2\frac{1+\b\g}{1+\b} b^+ + 2\b \frac{1+\b\g}{1+\b} b^-
\]
by \eqref{assab}.
The estimate of the next two terms in \eqref{estB} is based on \eqref{esttZ}, which gives
\[
\left|A(\b,\g) f''(0) \int_0^{t\wedge \tau^N_h} Y_N(s)Z_N(s) ds + B(\b,\g) f''(0) \int_0^{t\wedge \tau^N_h} Z_N^2(s) ds \right| \leq \mbox{constant} \cdot \int_0^t \tZ_N^2(s)\wedge  C^2(\b,\g) h^2 ds,
\]
where $C(\b,\g)$ is here the constant in \eqref{esttZ}.
This last term, by Lemma \ref{lemma:Z2} and dominated convergence, converges to zero in $L^1$ and so in probability. For the remaining two terms in \eqref{estB} we use again \eqref{boundZ},
which implies that
\be{estRem}
\int_0^{t \wedge \tau^N_h} R\left(\frac{1}{\sqrt{N}}Y_N\left(s\right) +  \frac{1}{\sqrt{N}}\frac{2\b}{1+\b}Z_N\left(s\right)  \right) \left(\frac{1}{\sqrt{N}}Y_N\left(s\right) +  \frac{1}{\sqrt{N}}\frac{2\b}{1+\b}Z_N\left(s\right)  \right)^3 ds \leq C
\ee
for some constant $C$. This implies that the first cubic term in \eqref{estB} goes to zero uniformly; the second term is similar and dealt with in the same way. This completes the proof of \eqref{eq:hyp4.6}. The proof of \eqref{eq:hyp4.7} is similar but easier, as it requires only a first order expansion for $f$
\end{proof}

\subsection{Convergence of the joint price-volatility process} \label{sec:price}
In this subsection we complete the proof of Theorem \ref{th:main}.
Let us recall the rescaled log-price in \eqref{def:pn}
\[
\begin{split}
\Pi_N(t)  = \frac{1}{\sqrt{N}} \sum_{i=1}^N P_i(\sqrt{N}t) & = \frac{1}{\sqrt{N}} \sum_{i=1}^N \left[ \int_{[0,\sqrt{N}t]\times [0,+\infty)} {\bf 1}_{[0,\l_N^{+}(s-))}(u) n_i^{+}(ds,du) \right. \\ & ~~~ \left. -  \int_{[0,\sqrt{N}t]\times [0,+\infty)} {\bf 1}_{[0,\l_N^{-}(s-))}(u) n_i^{-}(ds,du) \right],
\end{split}
\]
where, for the last equality, we use as in \eqref{macsde} the representation
\[
N_i^{\pm}(t) = \int_{[0,t]\times [0,+\infty)} {\bf 1}_{[0,\l_N^{\pm}(s-))}(u) \frac{1}{N} \sum_{i=1}^N n_i^{\pm}(ds,du).
\]
Compensating the Poisson Processes  as in \eqref{comppoint} we obtain
\be{price}
\begin{split}
\Pi_N(t)  & = N \left[ \int_0^{ t} f\left(\frac{1}{\sqrt{N}}Y_N\left(s\right) +  \frac{1}{\sqrt{N}}\frac{2\b}{1+\b}Z_N\left(s\right)  \right) ds \right. \\
& ~~~ - \left.   \int_0^{ t} f\left( \frac{1}{\sqrt{N}}Y_N(s) -  \frac{1}{\sqrt{N}} \frac{2}{1+\b}Z_N(s)  \right) ds \right] + M_N^{\Pi}(t),
\end{split}
\ee
where 
\[
M_N^{\Pi}(t) = \frac{1}{\sqrt{N}} \sum_{i=1}^N \left[ \int_{[0,\sqrt{N}t]\times [0,+\infty)} {\bf 1}_{[0,\l_N^{+}(s-))}(u) \tilde{n}_i^{+}(ds,du) -  \int_{[0,\sqrt{N}t]\times [0,+\infty)} {\bf 1}_{[0,\l_N^{-}(s-))}(u) \tilde{n}_i^{-}(ds,du) \right]
\]
 with $\tilde{n}_i^{\pm}(ds,du) = n_i^{\pm}(ds,du) - ds \, du$, is a martingale with zero mean and conditional quadratic variation 
\be{pricequadvar}
\sqrt{N} \left[ \int_0^{ t} f\left(\frac{1}{\sqrt{N}}Y_N\left(s\right) +  \frac{1}{\sqrt{N}}\frac{2\b}{1+\b}Z_N\left(s\right)  \right) ds + \int_0^{ t} f\left( \frac{1}{\sqrt{N}}Y_N(s) -  \frac{1}{\sqrt{N}} \frac{2}{1+\b}Z_N(s)  \right) ds \right] .
\ee
We can also easily identify the covariation between the martingales $M_N^{\Pi}$ and $M_N^{Y}$, which is given by
\begin{multline} \label{quadcovar}
\sqrt{N} \left[ \frac{1+\b\g}{1+\b}\int_0^t  f\left(\frac{1}{\sqrt{N}}Y_N\left(s\right) +  \frac{1}{\sqrt{N}}\frac{2\b}{1+\b}Z_N\left(s\right) \right) ds \right. \\  
\left.- \b \frac{1+\b\g}{1+\b} \int_0^t  f\left( \frac{1}{\sqrt{N}}Y_N(s) -  \frac{1}{\sqrt{N}} \frac{2}{1+\b}Z_N(s)  \right) ds \right].
\end{multline}
We can now rewrite \eqref{price} by using the Taylor expansion \eqref{taylor}, and we obtain
\be{priceexp}
\begin{split}
\Pi_N(t)  & = 2f'(0)\sqrt{N}\int_0^t Z_N(s) ds  \\ & ~~~ + \frac{1}{2} f''(0) \int_0^t \left[ \left( Y_N(s) + \frac{2\b}{1+\b} Z_N(s)\right)^2 - \left( Y_N(s) - \frac{2}{1+\b} Z_N(s)\right)^2 \right]ds
\\ & ~~~ + \frac{1}{\sqrt{N}} \int_0^t R\left( Y_N(s) + \frac{2\b}{1+\b} Z_N(s)\right)\left( Y_N(s) + \frac{2\b}{1+\b} Z_N(s)\right)^3 ds \\ &  ~~~ + \frac{1}{\sqrt{N}} \int_0^t R\left( Y_N(s) - \frac{2}{1+\b} Z_N(s)\right)\left( Y_N(s) - \frac{2}{1+\b} Z_N(s)\right)^3 ds
 + M_N^{\Pi}(t) \\ &  = 2f'(0)\sqrt{N}\int_0^t Z_N(s) ds  \\ & ~~~ +  f''(0) \int_0^t \left[  2Y_N(s) Z_N(s) + 2\frac{\b-1}{\b+1} Z_N^2(s) \right]ds
\\ & ~~~ + \frac{1}{\sqrt{N}} \int_0^t R\left( Y_N(s) + \frac{2\b}{1+\b} Z_N(s)\right)\left( Y_N(s) + \frac{2\b}{1+\b} Z_N(s)\right)^3 ds \\ &  ~~~ + \frac{1}{\sqrt{N}} \int_0^t R\left( Y_N(s) - \frac{2}{1+\b} Z_N(s)\right)\left( Y_N(s) - \frac{2}{1+\b} Z_N(s)\right)^3 ds
 + M_N^{\Pi}(t).
\end{split}
\ee
The problem here is in the term $2f'(0)\sqrt{N}\int_0^t Z_N(s) ds$. 
It follows from Lemma \ref{lemma:Z3} and the fact that $Z_N$ and $\tilde{Z}_N$ are equal up to time $\tau^N_h$, that
\[
\lim_{N \ra +\infty} \E \left[ \int_0^{T \wedge \tau^N_h} Z_N^2(s) ds \right] = 0,
\]
but this does not allow to control $\sqrt{N}\int_0^t Z_N(s) ds$.  To overcome this difficulty, we go back to \eqref{eqZ_N}; replacing once again $n_i^{\pm}(ds,du)$ by $\tilde{n}_i^{\pm}(ds,du) + ds \, du$ we obtain
\[
\begin{split}
Z_N(t) & = E_N + F_N t - \a \sqrt{N} \int_0^{t} Z_N(s) ds \\ & ~~~~~~+ \frac{1-\g}{2 }\sqrt{N} \int_0^{\sqrt{N}t}\left[ f\left(\frac{1}{\sqrt{N}}Y_N\left(s/\sqrt{N}\right) +  \frac{1}{\sqrt{N}}\frac{2\b}{1+\b}Z_N\left(s/\sqrt{N}\right) \right) \right. \\ & ~~~~~~~~~~~~~~~~~~~~~~ \left.- f\left( \frac{1}{\sqrt{N}}Y_N(s/\sqrt{N}) -  \frac{1}{\sqrt{N}} \frac{2}{1+\b}Z_N(s/\sqrt{N}) \right) \right]ds + M_N^Z(t) \\
& = E_N + F_N t - \a \sqrt{N} \int_0^{t} Z_N(s) ds \\ & ~~~~~~+ \frac{1-\g}{2 }N \int_0^{t}\left[ f\left(\frac{1}{\sqrt{N}}Y_N\left(s\right) +  \frac{1}{\sqrt{N}}\frac{2\b}{1+\b}Z_N\left(s\right) \right) \right. \\ & ~~~~~~~~~~~~~~~~~~~~~~ \left.- f\left( \frac{1}{\sqrt{N}}Y_N(s) -  \frac{1}{\sqrt{N}} \frac{2}{1+\b}Z_N(s) \right) \right]ds + M_N^Z(t)
\end{split}
\]
where $M_N^Z(t) = \frac{1-\g}{2} M_N^{\Pi}(t)$. Expanding $f$  as in \eqref{taylor} we obtain
\be{altZN}
\begin{split}
Z_N(t) & = E_N + F_N t - (\a- (1-\g) f'(0)) \sqrt{N} \int_0^{t} Z_N(s) ds \\ & ~~~~~~+ (1-\g) f''(0) \int_0^t \left[Y_N(s) Z_N(s) + \frac{\b-1}{1+\b} Z_N^2(s) \right] ds + Q_N(t) + \frac{1-\g}{2} M_N^{\Pi}(t),
\end{split}
\ee
where $Q_N(t)$ comes from the cubic part of the expansion;  by the same argument used in \eqref{estRem} it is shown that and $Q_N(t \wedge \tau_h^N)$ converges uniformly to zero as $N \ra +\infty$. We can now use \eqref{altZN} to eliminate the divergent term from \eqref{priceexp}.
\be{priceeq}
\begin{split}
\Pi_N(t) & = -\frac{2 f'(0)}{\a-(1-\g)f'(0)} Z_N(t) + \frac{2 f'(0)}{\a-(1-\g)f'(0)} E_N + \frac{2 f'(0)}{\a-(1-\g)f'(0)} F_N t \\ & + 2f''(0)\left( 2+ \frac{2 f'(0)(1-\g)}{\a-(1-\g)f'(0)} \right) \int_0^t \left[ Y_N(s)Z_N(s) + \frac{\b-1}{1+\b} Z_N^2(s)\right]ds \\ & + \left(1+\frac{f'(0)(1-\g)}{\a-(1-\g)f'(0)}\right) M_N^{\Pi}(t) + Q_N(t),
\end{split}
\ee
where $Q_N(t)$ is the same as the one in \eqref{altZN} up to a multiplicative constant.

We are now ready to use Theorem \ref{thm:diffusion:approx} for the two dimensional process $(\Pi_N, Y_N)$. The convergence of $Y_N$ has been proved in Proposition \ref{th:volatility}. The initial condition for the limit process $\pi$ follows from \eqref{priceeq} and Lemma \ref{lemma:Z3}, which guarantees that the term $-\frac{2 f'(0)}{\a-(1-\g)f'(0)} Z_N(t)$ is negligible for $t \in (0,T]$. 

So, we need to show the convergence of: the drift part of $\Pi_N$,  the  conditional quadratic variation of the martingale part of $\Pi_N$, and the quadratic covariation of the two martingale components. 

The drift part of $\Pi_N$ contains the term
\[
\frac{2 f'(0)}{\a-(1-\g)f'(0)} F_N t
\]
which asymptotically produces the constant drift 
\[
\frac{f'(0)(1-\g)}{\a-(1-\g)f'(0)}(b^+ - b^-).
\]
To complete the treatment of the drift we need to show that, for every $\e,\d>0$,
\[
\lim_{N \ra +\infty} \P \left(\sup_{\e \leq t \leq T \wedge \tau_h^N} \left|\int_0^t \left[ Y_N(s)Z_N(s) + \frac{\b-1}{1+\b} Z_N^2(s)\right]ds \right| \geq \delta \right) = 0,
\]
which readily follows from Lemma \ref{lemma:Z3} and the fact that $Y_N$ is bounded up to time $\tau_h^N$.

The martingale component of $\Pi_N$ is given by, using $\a = f'(0)(1+\b\g)$,
\[
\left(1+\frac{f'(0)(1-\g)}{\a-(1-\g)f'(0)}\right) M_N^{\Pi}(t) = \frac{1+\b\g}{\g(1+\b)}M_N^{\Pi}(t).
\]
Thus, using \eqref{pricequadvar}, we need to show that
\be{estmartprice}
\lim_{N \ra +\infty} \P \left(\sup_{\e \leq t \leq T \wedge \tau_h^N} \left|\int_0^t A_N(s) ds - \int_0^t 2 f'(0) Y_N(s) ds \right| \geq \delta \right) = 0,
\ee
where
\[
A_N(s) := \sqrt{N} \left[f\left( \frac{1}{\sqrt{N}}Y_N(s) -  \frac{1}{\sqrt{N}} \frac{2}{1+\b}Z_N(s)  \right) + f\left( \frac{1}{\sqrt{N}}Y_N(s) -  \frac{1}{\sqrt{N}} \frac{2}{1+\b}Z_N(s)\right)\right],
\]
which is established by a simple first order Taylor expansion of $f$.

Finally we need to verify the convergence of the covariation. Using \eqref{quadcovar}, it follows that the covariation $\langle  M^{\Pi}_N, M_{Y}^N \rangle(t)$ converges, in the sense of \eqref{eq:hyp4.7}, to
\[
f'(0) \frac{(1-\b)(1+\b \g)}{1+\b} \int_0^t y(s) ds.
\]
Using the specific values of $\sigma_{\pi}$ and $\sigma_{y}$ in \eqref{volmod}, we obtain
\[
\frac{\sqrt{f'(0) y(t) (1+\b^2)}(1+\b\g)}{1+\b} \frac{1+\b\g}{\g(1+\b)}\sqrt{2 f'(0) y(t)} d\langle B,W\rangle (t) = f'(0) \frac{(1-\b)(1+\b \g)}{1+\b} y(t) dt,
\]
which implies \eqref{intro:limcorr} and completes the proof. \\

We remark that this convergence of the price process $\Pi_N(t)$ for $t \in (0,T]$ does not extend to $t=0$; indeed the term $-\frac{2 f'(0)}{\a-(1-\g)f'(0)} Z_N(t)$ in \eqref{priceeq} does not go to zero unless $a^+ = a^-$.

\subsection{Proof of Theorem \ref{th:maininhom}}

The proof follows the same arguments as that of Theorem \ref{th:main}. Proceeding step by step, we illustrate the changes needed to cover the inhomogeneous case. For readability of the formulas, we omit the index $N$ in all empirical averages $\overline{f(\b,\g)}_N$; their formal identification to their limit values $\overline{f(\b,\g)}$ does not causes inconsistencies.

\bigskip

\noindent
{\em Step 1: representation of $m_N^{\pm}$ (Lemma \ref{lemma:mN})}

\smallskip
\noindent
The following modification of Lemma \ref{lemma:mN} clearly holds: 
\begin{equation} \label{macsdeinhom}
\begin{split}
m_N^{\pm}(t) &  = - \a \int_0^t m_N^{\pm}(s) ds  + \frac{1}{N} \sum_{i=1}^N \int_0^t \l_N^{i,\pm}(s)ds + a^{\pm}_N + b^{\pm}_N t+ M_N^{\pm}(t) ,
\end{split}
\end{equation}
where $\l_N^{i,\pm}$ are given in \eqref{lambdai} 
\[
M_N^{\pm}(t) := \int_{[0,t]\times [0,+\infty)} {\bf 1}_{[0,\l_N^{\pm}(s-))}(u) \frac{1}{N} \sum_{i=1}^N \tilde{n}_i^{\pm}(ds,du)
\]
are orthogonal martingales with conditional quadratic variations $\frac{1}{N^2}\sum_i \int_0^t \l_N^{i,\pm}(s)ds$. 

\bigskip
\noindent
{\em Step 2: definition of the volatility process and collapsing of $Z_N$ (Lemmas \ref{lemma:Z1}, \ref{lemma:Z2} and \ref{lemma:Z3})}

\smallskip
\noindent

We now define
\[
Y_N(t) 
 = \sqrt{N} \left[\frac{\overline{\g}(1+\overline{\b\g})}{\overline{\g}+ \overline{\b\g}} m_N^+(\sqrt{N}t) + \frac{\overline{\b\g}(1+\overline{\b\g})}{\overline{\g}+ \overline{\b\g}} m_N^-(\sqrt{N}t)\right]
\]
and
\[
Z_N(t) 
 = \frac{(1-\overline{\g})\sqrt{N}}{2} \left[  m_N^+(\sqrt{N}t) -  m_N^-(\sqrt{N}t)\right].
\]
Note that
\[
\l_N^{+}(t) = f \left( \frac{1}{\sqrt{N}}\frac{1+\b^N_i\g^N_i}{1+\overline{\b\g}}Y_N(t/\sqrt{N}) +  \frac{2}{\sqrt{N}} \frac{\overline{\b\g} - \b^N_i\g^N_i \overline{\g}}{(1-\overline{\g})(\overline{\b\g}+ \overline{\g})}Z_N(t/\sqrt{N}) \right)
\]
and
\[
\l_N^{-}(t) = f \left( \frac{1}{\sqrt{N}}\frac{1+\b^N_i\g^N_i}{1+\overline{\b\g}}Y_N(t/\sqrt{N}) +  \frac{2}{\sqrt{N}} \frac{\g^N_i \overline{\b\g} - \overline{\g} -\b^N_i\g^N_i \overline{\g} + \g^N_i \overline{\g} }{(1-\overline{\g})(\overline{\b\g}+ \overline{\g})}Z_N(t/\sqrt{N}) \right)
\]
We obtain the generalization of \eqref{eqZ_N}:
\be{eqtZ_Ninhom}
\begin{split}
Z_N(t) & = E_N + F_N t  - \a \sqrt{N} \int_0^{t} Z_N(s) ds \\ & ~~~~~~+ \frac{1-\overline{\g}}{2 \sqrt{N}} \sum_{i=1}^N \int_{[0,\sqrt{N} t]} {\bf 1}_{[0,\l_N^{i,+}(s))}(u)  n_i^+(ds,du) - \frac{1-\overline{\g}}{2 \sqrt{N}} \sum_{i=1}^N\int_{[0,\sqrt{N} t]} {\bf 1}_{[0,\l_N^{i,-}(s))}(u) \sum_{i=1}^N n_i^-(ds,du),
\end{split}
\ee
where 
\[
E_N = \frac{1-\overline{\g}}{2} \sqrt{N} [a_N^+ - a_N^-] \ \ \ F_N = \frac{1-\overline{\g}}{2} \sqrt{N} [b_N^+ - b_N^-].
\]
After having defined the truncated processes $\tY_N$ and $\tZ_N$, the proof of the result in Lemma \ref{lemma:Z2} comes again from the semimartingale representation of $\tZ_N^4$, that now reads
\be{tZ2martinhom}
\begin{split}
\tZ_N^4(t)  & = E_N^4 + 4 F_N \int_0^t \tZ^2_N(s)ds - 4 \a \sqrt{N} \int_0^t \tZ_N^4(s)ds \\ & + \sqrt{N} \int_0^{t}\left[\frac{(1-\overline{\g})^4}{16N} + \frac{(1-\overline{\g})^3}{ 2 N^{\frac32}}\tZ_N\left(s\right) + 3\frac{(1-\overline{\g})^2}{ 2 N}\tZ^2_N\left(s\right) + 2\frac{(1-\overline{\g})}{ \sqrt{N}}\tZ^3_N\left(s\right)\right] \\ & ~~~~~~~~~~~~~ \sum_{i=1}^N f \left( \frac{1}{\sqrt{N}}\frac{1+\b^N_i\g^N_i}{1+\overline{\b\g}}\tY_N(s) +  \frac{2}{\sqrt{N}} \frac{\overline{\b\g} - \b^N_i\g^N_i \overline{\g}}{(1-\overline{\g})(\overline{\b\g}+ \overline{\g})}\tZ_N(s) \right) ds \\ & + \sqrt{N} \int_0^{t}\left[\frac{(1-\overline{\g})^4}{16N} - \frac{(1-\overline{\g})^3}{ 2 N^{\frac32}}\tZ_N\left(s\right) + 3\frac{(1-\overline{\g})^2}{ 2 N}\tZ^2_N\left(s\right) - 2\frac{(1-\overline{\g})}{ \sqrt{N}}\tZ^3_N\left(s\right)\right] \\ & ~~~~~~~~~~~~~ \sum_{i=1}^N f \left( \frac{1}{\sqrt{N}}\frac{1+\b^N_i\g^N_i}{1+\overline{\b\g}}\tY_N(s) +  \frac{2}{\sqrt{N}} \frac{\g^N_i \overline{\b\g} - \overline{\g} -\b^N_i\g^N_i \overline{\g} + \g^N_i \overline{\g} }{(1-\overline{\g})(\overline{\b\g}+ \overline{\g})}\tZ_N(s) \right) ds \\ & + M_N^{\tZ^4}(t).
\end{split}
\ee
By using the fact that $f$ is concave and Lipschitz, as in \eqref{conc} and \eqref{Lip}, we have
\begin{multline*}
\sum_{i=1}^N \left[f \left( \frac{1}{\sqrt{N}}\frac{1+\b^N_i\g^N_i}{1+\overline{\b\g}}\tY_N(s) +  \frac{2}{\sqrt{N}} \frac{\overline{\b\g} - \b^N_i\g^N_i \overline{\g}}{(1-\overline{\g})(\overline{\b\g}+ \overline{\g})}\tZ_N(s) \right)  \right. \\ + \left.f \left( \frac{1}{\sqrt{N}}\frac{1+\b^N_i\g^N_i}{1+\overline{\b\g}}\tY_N(s) +  \frac{2}{\sqrt{N}} \frac{\g_i \overline{\b\g} - \overline{\g} -\b^N_i\g^N_i \overline{\g} + \g^N_i \overline{\g} }{(1-\overline{\g})(\overline{\b\g}+ \overline{\g})}\tZ_N(s) \right)\right]\\
\leq f'(0) \sum_{i=1}^N\left[ \left( \frac{1}{\sqrt{N}}\frac{1+\b^N_i\g^N_i}{1+\overline{\b\g}}\tY_N(s) +  \frac{2}{\sqrt{N}} \frac{\overline{\b\g} - \b^N_i\g^N_i \overline{\g}}{(1-\overline{\g})(\overline{\b\g}+ \overline{\g})}\tZ_N(s) \right) \right. \\+ \left.\left( \frac{1}{\sqrt{N}}\frac{1+\b^N_i\g^N_i}{1+\overline{\b\g}}\tY_N(s) +  \frac{2}{\sqrt{N}} \frac{\g_i \overline{\b\g} - \overline{\g} -\b^N_i\g^N_i \overline{\g} + \g^N_i \overline{\g} }{(1-\overline{\g})(\overline{\b\g}+ \overline{\g})}\tZ_N(s) \right) \right] \\ \leq \mbox{constant} \cdot\sqrt{N} (\tY_N(s) + |\tZ_N(s)|)
\end{multline*}
and
\begin{multline*}
\sum_{i=1}^N \left|f \left( \frac{1}{\sqrt{N}}\frac{1+\b^N_i\g^N_i}{1+\overline{\b\g}}\tY_N(s) +  \frac{2}{\sqrt{N}} \frac{\overline{\b\g} - \b^N_i\g^N_i \overline{\g}}{(1-\overline{\g})(\overline{\b\g}+ \overline{\g})}\tZ_N(s) \right)  \right. \\ - \left.f \left( \frac{1}{\sqrt{N}}\frac{1+\b^N_i\g^N_i}{1+\overline{\b\g}}\tY_N(s) +  \frac{2}{\sqrt{N}} \frac{\g^N_i \overline{\b\g} - \overline{\g} -\b^N_i\g^N_i \overline{\g} + \g^N_i \overline{\g} }{(1-\overline{\g})(\overline{\b\g}+ \overline{\g})}\tZ_N(s) \right)\right|\\ \leq f'(0) \sum_{i=1}^N \left| \left( \frac{1}{\sqrt{N}}\frac{1+\b^N_i\g^N_i}{1+\overline{\b\g}}\tY_N(s) +  \frac{2}{\sqrt{N}} \frac{\overline{\b\g} - \b^N_i\g^N_i \overline{\g}}{(1-\overline{\g})(\overline{\b\g}+ \overline{\g})}\tZ_N(s) \right) \right. \\ \left.-  \left( \frac{1}{\sqrt{N}}\frac{1+\b^N_i\g^N_i}{1+\overline{\b\g}}\tY_N(s) +  \frac{2}{\sqrt{N}} \frac{\g^N_i \overline{\b\g} - \overline{\g} -\b^N_i\g^N_i \overline{\g} + \g^N_i \overline{\g} }{(1-\overline{\g})(\overline{\b\g}+ \overline{\g})}\tZ_N(s) \right) \right| \\ \leq  2f'(0)\cdot\sqrt{N} |\tZ_N(s)|,
\end{multline*}
In the first estimate the generic constant in the last term may depend on the values of all parameters; in the second estimate there are several crucial cancellations to obtain the specific upper bound.

Using these estimates, inequality \eqref{tZ2derest2} can be derived  and Lemma \ref{lemma:Z2} follows for this inhomogeneous model. The proof of Lemma \ref{lemma:Z3}, which uses a semimartingale representation for $\tZ_N^2$, follows along the same lines, and we omit the details.

\bigskip
\noindent
{\em Step 3: dynamics of the volatility (Proposition \ref{th:volatility})}

\smallskip
\noindent
We begin by giving, as in \eqref{Y}, the semimartingale representation of $Y_N$:
\be{Yinhom}
\begin{split}
Y_N(t) & = C_N + D_N t  -\a \sqrt{N} \int_0^t Y_N(s)  ds  \\ & ~~~+\frac{\overline{\g}(1+\overline{\b\g})}{\overline{\g}+ \overline{\b\g}} \sum_{i=1}^N \int_0^t f \left( \frac{1}{\sqrt{N}}\frac{1+\b^N_i\g^N_i}{1+\overline{\b\g}}\tY_N(s) +  \frac{2}{\sqrt{N}} \frac{\overline{\b\g} - \b^N_i\g^N_i \overline{\g}}{(1-\overline{\g})(\overline{\b\g}+ \overline{\g})}\tZ_N(s) \right) ds \\ & ~~~+\frac{\overline{\b\g}(1+\overline{\b\g})}{\overline{\g}+ \overline{\b\g}} \sum_{i=1}^N \int_0^t f \left( \frac{1}{\sqrt{N}}\frac{1+\b^N_i\g^N_i}{1+\overline{\b\g}}\tY_N(s) -  \frac{2}{\sqrt{N}}\frac{\g^N_i \overline{\b\g} - \overline{\g} -\b^N_i\g^N_i \overline{\g} + \g^N_i \overline{\g} }{(1-\overline{\g})(\overline{\b\g}+ \overline{\g})}\tZ_N(s) \right) ds \\ & ~~~ + M_N^Y(t),
\end{split}
\ee
where 
\[
\begin{split}
C_N & = \sqrt{N} \left[  \frac{\overline{\g}(1+\overline{\b\g})}{\overline{\g}+ \overline{\b\g}} a_N^+ +  \frac{\overline{\b\g}(1+\overline{\b\g})}{\overline{\g}+ \overline{\b\g}} a_N^- \right] \\
D_N & = \sqrt{N} \left[  \frac{\overline{\g}(1+\overline{\b\g})}{\overline{\g}+ \overline{\b\g}} b_N^+ +  \frac{\overline{\b\g}(1+\overline{\b\g})}{\overline{\g}+ \overline{\b\g}} b_N^- \right]
\end{split}
\]
and
$M_N^Y(t)$ is a martingale with conditional quadratic variation
\[
\begin{split}
\left(\frac{\overline{\g}(1+\overline{\b\g})}{\overline{\g}+ \overline{\b\g}}\right)^2 \frac{1}{\sqrt{N}} \sum_{i=1}^N \int_0^t f \left( \frac{1}{\sqrt{N}}\frac{1+\b^N_i\g^N_i}{1+\overline{\b\g}}Y_N(s) +  \frac{2}{\sqrt{N}} \frac{\overline{\b\g} - \b^N_i\g^N_i \overline{\g}}{(1-\overline{\g})(\overline{\b\g}+ \overline{\g})}Z_N(s) \right) ds \\ + \left( \frac{\overline{\b\g}(1+\overline{\b\g})}{\overline{\g}+ \overline{\b\g}}\right)^2 \frac{1}{\sqrt{N}} \sum_{i=1}^N \int_0^t f \left( \frac{1}{\sqrt{N}}\frac{1+\b^N_i\g^N_i}{1+\overline{\b\g}}Y_N(s) -  \frac{2}{\sqrt{N}}\frac{\g^N_i \overline{\b\g} - \overline{\g} -\b^N_i\g^N_i \overline{\g} + \g^N_i \overline{\g} }{(1-\overline{\g})(\overline{\b\g}+ \overline{\g})}Z_N(s) \right) ds.
\end{split}
\]
From this point on one proceeds with the second order expansion of $f$ as in the proof of Proposition \ref{th:volatility}; no appreciable difference in the rest of the proof emerges.

\bigskip
\noindent
{\em Step 4: dynamics of the price}

\medskip
\noindent

The price representation \eqref{price} now becomes 
\be{priceinhom}
\begin{split}
\Pi_N(t)  & = \sum_{i=1}^N  \left[\int_0^tf \left( \frac{1}{\sqrt{N}}\frac{1+\b^N_i\g^N_i}{1+\overline{\b\g}}\tY_N(s) +  \frac{2}{\sqrt{N}} \frac{\overline{\b\g} - \b^N_i\g^N_i \overline{\g}}{(1-\overline{\g})(\overline{\b\g}+ \overline{\g})}\tZ_N(s) \right)ds  \right. \\ & - \left. \int_0^tf \left( \frac{1}{\sqrt{N}}\frac{1+\b^N_i\g^N_i}{1+\overline{\b\g}}\tY_N(s) -  \frac{2}{\sqrt{N}} \frac{\g^N_i \overline{\b\g} - \overline{\g} -\b^N_i\g^N_i \overline{\g} + \g^N_i \overline{\g} }{(1-\overline{\g})(\overline{\b\g}+ \overline{\g})}\tZ_N(s) \right) ds \right] \\ & + M_N^{\Pi}(t)
\end{split}
\ee
where $M_N^{\Pi}$ is a martingale with zero mean and conditional quadratic variation 
\be{pricequadvarinhom}
\begin{split}
\sum_{i=1}^N & \left[\int_0^tf \left( \frac{1}{\sqrt{N}}\frac{1+\b^N_i\g^N_i}{1+\overline{\b\g}}\tY_N(s) +  \frac{2}{\sqrt{N}} \frac{\overline{\b\g} - \b^N_i\g^N_i \overline{\g}}{(1-\overline{\g})(\overline{\b\g}+ \overline{\g})}\tZ_N(s) \right)ds  \right. \\ & + \left. \int_0^tf \left( \frac{1}{\sqrt{N}}\frac{1+\b^N_i\g^N_i}{1+\overline{\b\g}}\tY_N(s) -  \frac{2}{\sqrt{N}} \frac{\g^N_i \overline{\b\g} - \overline{\g} -\b^N_i\g^N_i \overline{\g} + \g^N_i \overline{\g} }{(1-\overline{\g})(\overline{\b\g}+ \overline{\g})}\tZ_N(s) \right) ds \right] \end{split}
\ee
The quadratic covariation between $M_N^{\Pi}$ and $M_N^{Y}$ is now given by 
\be{quadcovarinhom}
\begin{split}
\frac{1}{\sqrt{N}} \sum_{i=1}^N \left[ \right. & \left. \frac{\overline{\g}(1+\overline{\b\g})}{\overline{\g}+\overline{\b\g}} \int_0^t f \left( \frac{1}{\sqrt{N}}\frac{1+\b_i\g_i}{1+\overline{\b\g}}\tY_N(s) +  \frac{2}{\sqrt{N}} \frac{\overline{\b\g} - \b_i\g_i \overline{\g}}{(1-\overline{\g})(\overline{\b\g}+ \overline{\g})}\tZ_N(s) \right)ds  \right. 
\\ & + \left.  \frac{\overline{\b\g}(1+\overline{\b\g})}{\overline{\g}+\overline{\b\g}}\int_0^tf \left( \frac{1}{\sqrt{N}}\frac{1+\b_i\g_i}{1+\overline{\b\g}}\tY_N(s) -  \frac{2}{\sqrt{N}} \frac{\g_i \overline{\b\g} - \overline{\g} -\b_i\g_i \overline{\g} + \g_i \overline{\g} }{(1-\overline{\g})(\overline{\b\g}+ \overline{\g})}\tZ_N(s) \right) ds \right]
\end{split}
\ee
Taylor expansion of $f$ in \eqref{priceinhom} leads, as in the proof in Section \ref{sec:price}, to a representation of the form
\be{priceexpinhom}
\Pi_N(t) = 2 f'(0) \sqrt{N} \int_0^t Z_N(s) ds + \mbox{ vanishing terms } + M_N^{\Pi}(t).
\ee
To deal with the divergent term, as in Section \ref{sec:price}, we derive from \eqref{eqtZ_Ninhom}
\[
Z_N(t) = E_N + F_N t - (\a-(1-\overline{\g})f'(0)) \sqrt{N} \int_0^t Z_N(s) ds +  \mbox{ vanishing terms } + \frac{1-\overline{\g}}{2}M_N^{\Pi}(t),
\]
that we solve for $\sqrt{N} \int_0^t Z_N(s) ds$ and replace into \eqref{priceexpinhom}. No other relevant difference emerges for the proof in Section \ref{sec:price}.

\subsection{Proof of Theorem \ref{th:mainself}}

The extension of the proof to this case is rather straightforward, so we only discuss the key points. To begin with it is convenient to define
\[
Y^i_N(t) 
 = \sqrt{N} \left[\frac{1+\b\g}{1+\b} X_i^+(\sqrt{N}t) + \frac{\b(1+\b\g)}{1+\b} X_i^-(\sqrt{N}t)\right]
\]
and
\[
Z^i_N(t) 
 = \frac{(1-\g)\sqrt{N}}{2} \left[  X_i^+(\sqrt{N}t) -  X_i^-(\sqrt{N}t)\right],
\]
so that
\[
Y_N(t) = \frac{1}{N} \sum_{i=1}^N Y^i_N(t) \hspace{1cm} Z_N(t) = \frac{1}{N} \sum_{i=1}^N Z^i_N(t).
\]
The representation \eqref{eqZ_N} still holds, but now in the form
\be{selfnew1}
\begin{split}
Z_N(t) & = E_N + F_N t  - \a \sqrt{N} \int_0^{t} Z_N(s) ds \\ & ~~~~~~+ \frac{1-\g}{2 \sqrt{N}} \sum_{i=1}^N \int_{[0,\sqrt{N} t]} {\bf 1}_{[0,\l_N^{i,+}(s))}(u)  n_i^+(ds,du) - \frac{1-\g}{2 \sqrt{N}} \sum_{i=1}^N\int_{[0,\sqrt{N} t]} {\bf 1}_{[0,\l_N^{i,-}(s))}(u) \sum_{i=1}^N n_i^-(ds,du),
\end{split}
\ee
where
\be{selfnew2}
\l_N^{i,+}(t) = f \left( \frac{1}{\sqrt{N}}\left[Y_N(t/\sqrt{N}) + \frac{\kappa}{\sqrt{N}} Y^i_N(t/\sqrt{N}) \right] +  \frac{1}{\sqrt{N}} \frac{2\b}{1+\b}\left[Z_N(t/\sqrt{N}) + \frac{\kappa}{\sqrt{N}} Z^i_N(t/\sqrt{N}) \right] \right)
\ee
and
\be{selfnew3}
\l_N^{i,-}(t) = f \left( \frac{1}{\sqrt{N}}\left[Y_N(t/\sqrt{N}) + \frac{\kappa}{\sqrt{N}} Y^i_N(t/\sqrt{N}) \right]  -  \frac{1}{\sqrt{N}} \frac{2}{1+\b}\left[Z_N(t/\sqrt{N}) + \frac{\kappa}{\sqrt{N}} Z^i_N(t/\sqrt{N}) \right] \right).
\ee
The proof of collapsing of $Z_N$ requires no changes: it only uses the first order properties of $f$, so the extra terms $\frac{\kappa}{\sqrt{N}} Y^i_N(t/\sqrt{N}) $ and $\frac{\kappa}{\sqrt{N}} Z^i_N(t/\sqrt{N}) $ give a negligible contribution as $N \rightarrow +\infty$. The only difference emerges in the analysis of the volatility, which starts from
\be{Yself}
\begin{split}
Y_N(t) & = C_N + D_N t  -\a \sqrt{N} \int_0^t Y_N(s)  ds  \\ & ~~~+ \sum_{i=1}^N \left[ \frac{1+\b\g}{1+\b} \int_0^{ t} f \left( \frac{1}{\sqrt{N}}\left[Y_N(s) + \frac{\kappa}{\sqrt{N}} Y^i_N(s) \right]  \right. \right. \\ & ~~~~~~~~~~~~~~~~~~~~~~~~~~~~~~ +  \left. \left. \frac{1}{\sqrt{N}} \frac{2\b}{1+\b}\left[Z_N(s) + \frac{\kappa}{\sqrt{N}} Z^i_N(s) \right] \right) ds \right. \\
& ~~~ + \left. \b\frac{1+\b\g}{1+\b}  \int_0^{ t} f \left( \frac{1}{\sqrt{N}}\left[Y_N(s) + \frac{\kappa}{\sqrt{N}} Y^i_N(s) \right]  \right. \right. \\ & ~~~~~~~~~~~~~~~~~~~~~~~~~~~~~~ - \left. \left.  \frac{1}{\sqrt{N}} \frac{2}{1+\b}\left[Z_N(s) + \frac{\kappa}{\sqrt{N}} Z^i_N(s) \right] \right) ds \right]  \\ & ~~~ + M_N^Y(t)
\end{split}
\ee
As before, criticality condition \eqref{intro:criticalself} only depends on dominant terms in the first order expansion of $f$, to which the new extra terms do not contribute. However, the term of the first order expansion $\frac{\kappa}{N} f'(0) \sum_{i=1}^N Y^i_N(s) = \kappa f'(0) Y_N(s)$ has order one, the same as the quadratic term containing $Y^2_N(s)$. This explains the extra linear term in the drift of the volatility in \eqref{volmodself}. The rest of the proof requires no change as, again, the derivation of the price dynamics only depends on the linear terms in the expansion of $f$.

\begin{remark} \label{rem:nonmar}
The representations \eqref{selfnew1}, \eqref{selfnew2} and \eqref{selfnew3} imply that the pair $(Z_N,Y_N)$ is not a Markov process. To see this we observe that \eqref{micsde}, with $\l^{i, \pm}$ in place of $\l^{\pm}$ implies that the $[0,+\infty)^{N}$-valued process $\mathbf{X}^{\pm}(t)  = (X_i^{\pm}(t))_{i=1}^N$ is a Markov process whose infinitesimal generator $\mathcal{L}_N$ acts on a smooth $F:[0,+\infty)^{N} \ra \R$ as follows: for $\mathbf{x} = (x_1^{\pm}, \ldots , x_N^{\pm})$
\[
\mathcal{L}_N F(\mathbf{x})  = \sum_{s \in \{\pm\}} \sum_{i=1}^N \left(-\a x_i^s + b_N^s]\right) \frac{\partial F}{\partial x_i^s} (\mathbf{x})  +  \sum_{s \in \{\pm\}} \sum_{i=1}^N \l_N^{s,i}(\mathbf{x}) \left[ T_i^s F(\mathbf{x}) - F(\mathbf{x}) \right],
\]
where $T_i^s F(\mathbf{x})$ is obtained from $F(\mathbf{x})$ by replacing $x_i^s$ with $x_i^s +1$ and, by slight abuse of notation,
\[
\begin{split}
\l_N^{i,+}(\mathbf{x}) &  := f(m_N^{i,+}(\mathbf{x}) + \b \g m_N^{i,-}(\mathbf{x})) \\
\l_N^{i,-}(\mathbf{x}) & := f( \g m_N^{i,+}(\mathbf{x}) + (1+(\b-1) \g) m_N^{i,-}(\mathbf{x})),
\end{split}
\]
with 
\[
m_N^{i,\pm}(\mathbf{x}) = \frac{1}{N} \sum_{j=1}^N x_j^{\pm} + \frac{\kappa}{\sqrt{N}} x_i^{\pm}.
\]
Note that if $(Z_N(t),Y_N(t))$ were a Markov process, then $m_N^{i,\pm}(\mathbf{X(t)})$ would also be a Markov process, since they are linked by an invertible linear transformation (except in the special case $\b = \g = 1$; but in this case $Z_N = 0$ and the argument can be adapted by replacing the pair $m_N^{i,\pm}(\mathbf{X(t)})$ with the sum of the two components). This would imply that, for a smooth $\Psi:[0,+\infty)^2 \ra \R$, the generator
$
\mathcal{L}_N \Psi(m_N^{\pm}(\mathbf{x}))
$
would be a function of $m_N^{\pm}(\mathbf{x})$. This is not the case. Indeed observe that, by permutation invariance, $T_i^s \Psi - \Psi$ does not depend on $i$, and therefore in $\mathcal{L}_N \Psi(m_N^{\pm}(\mathbf{x}))$ we obtain the term
$
\sum_{i=1}^N \l_N^{i,s}(\mathbf{x})
$
which is not a function of $m_N^{\pm}(\mathbf{x})$ unless $f$ is linear.

\end{remark}

\appendix

\section{Useful results}

These theorems correspond to Theorem VII, 4.1 \cite{EtKu2005} and the Proposition in Appendix A in \cite{CE88}.

\begin{theorem}[\textbf{Diffusion approximation }]\label{thm:diffusion:approx}
Let $\left(X_N\right)_{N\in \mathbb{N}}$ and  $\left(B_N\right)_{N\in \mathbb{N}}$ be $\R^d$-valued processes with càdlàg sample paths and let $A_N = \left(A_N^{i j}\right)$ be a symmetric $d\times d$ matrix-valued process such that $A_N^{i j}$ has càdlàg sample paths in $\R$ and $A_N(t)-A_N(s)$ is non-negative definite for all $t> s \geq 0$. 
Let $\mathcal{F}^N_t \coloneqq \sigma\left(X_N(s), B_N(s), A_N(s)\, : \, s\leq t \right)$. 
Let $\tau^N_h \coloneqq \inf\left\{t > 0\, : \, \vert X_N(t)\vert\geq h \, \text{or}\, \vert X_N(t^-)\vert \geq h \right\}$, where $X_N(t^-) \coloneqq \lim_{s\to t^-} X_N(s)$ denotes the left-hand limit. 
Assume that
\begin{itemize}
    \item $M_N \coloneqq X_N - B_N$ and $M_N^i M_N^j - A_N^{i j }$ $i,j=1,\dots,d$ are $\left(\mathcal{F}^N_t\right)$-local martingales
    \item for each $T>0$, $h>0$
    \begin{equation}\label{eq:hyp4.4}
        \lim_{N \to +\infty}\E\left[ \sup_{t \leq T \wedge \tau^N_h} \vert B_N(t) - B_N(t^-)\vert^2\right] = 0
    \end{equation}
    \item for each $T>0$, $h>0$, $i,j=1,\dots,d$
    \begin{equation}\label{eq:hyp4.5}
        \lim_{N \to +\infty}\E\left[ \sup_{t \leq T \wedge \tau^N_h} \vert A_N^{i j}(t) - A_N^{i j}(t^-)\vert\right] = 0 
    \end{equation}
    \item for each $T>0$, $h>0$ 
    \begin{equation}\label{eq:hyp4.3}
        \lim_{N \to +\infty}\E\left[ \sup_{t \leq T \wedge \tau^N_h} \vert X_N(t) - X_N(t^-)\vert^2\right] = 0 
    \end{equation}
    \item there exist a continuous, symmetric, non-negative definite $d\times d$ matrix-valued function on $\R^d$, $a=\left(a_{i j}\right)$, and a continuous function $b: \R^d \to \R^d$ such that, for each $h>0$, $T>0$ and $i,j=1,\dots,d$, and for all $\epsilon>0$, 
    \begin{equation}\label{eq:hyp4.7}
        \lim_{N\to +\infty}\mathbb{P}\left(\sup_{t\leq T\wedge \tau^N_h} \Bigg\vert A_N^{ i j}(t) - \int_0^t a_{i j }(X_N(s))ds \Bigg\vert > \epsilon \right) = 0 
    \end{equation}
    and 
    \begin{equation}\label{eq:hyp4.6}
        \lim_{N\to +\infty}\mathbb{P}\left(\sup_{t\leq T \wedge \tau^N_h} \Bigg\vert B_N^i(t) - \int_0^t b_i(X_N(s))ds \Bigg\vert > \epsilon \right) = 0 
    \end{equation}
    \item the $\mathcal{C}_{\R^d}\left([0,+\infty)\right)$ martingale problem for
	\begin{equation}\label{A:EtKu}
	\tilde{A} \coloneqq \left\{\left(f,Gf\coloneqq \frac{1}{2}\sum_{i,j} a_{i j}\partial_i \partial_j f + \sum_i b_i \partial_i f\right)\, : \, f \in \mathcal{C}^{\infty}_c\left(\R^d\right) \right\} 
	\end{equation}
is well-posed. 

    \item the sequence of the initial laws of the $X_N$'s converges in distribution to some probability distribution on $\R^d$, $\nu$. 
\end{itemize}
Then $\left(X_N\right)_N$ converges in distribution to the solution of the martingale problem for $(\tilde{A},\nu)$. 
That is, the laws of the processes $\left(X_N\right)_N$ converge weakly to the law of a process $X$ which is a weak solution of the SDE 
\begin{equation}
    \begin{split}
        dX(t) = b(X(t))dt +\Sigma(X(t)) dW(t)
    \end{split}
\end{equation}
where $b = (b_i)_i$ and $\left(\Sigma \Sigma^T\right)_{i j }= a_{i j} $ are the drift vector and the diffusion coefficient in \eqref{A:EtKu}.
\end{theorem}

\begin{theorem} \label{collapsing}
Let $\{X_n(t)\}_{n \geq 1}$ be a sequence of positive semimartingales on a probability space $(\Omega, \mathcal{A}, \mathcal{P})$, with
\[dX_n(t) = S_n(t)dt + \int_{[0,t]\times \mathcal{Y}} f_n(s-, y) [\Lambda_n(ds,dy) - A_n(s,dy)ds]\]
Here, $\Lambda_n$ is a Point Process of intensity $A_n(t,dy)dt$ on $\mathbb{R}^+\times \mathcal{Y}$, where $\mathcal{Y}$ is a measurable space, and $S_n(t)$ and $f_n(t)$ are $\mathcal{A}_t$-adapted processes, if we consider $(\mathcal{A}_t)_{t\geq 0}$ a filtration on $(\Omega, \mathcal{A}, \mathcal{P})$ generated by $\Lambda_n$.\\ 
Let $d>1$ and $C_i$ constants independent of $n$ and $t$. Suppose there exist $\{\alpha_n\}_{n \geq1}$ and $\{\beta_n\}_{n \geq 1}$, increasing sequences with
\begin{equation} \label{condc1} n^{1/d}\alpha_n^{-1} \xrightarrow{n\rightarrow +\infty} 0, n^{-1} \alpha_n  \xrightarrow{n\rightarrow +\infty} 0, n^{-1}\beta_n  \xrightarrow{n\rightarrow +\infty} 0  \end{equation}
and
\begin{equation} \label{condc2} E \Big[ \Big( X_n(0) \Big)^d \Big] \leq C_1 \alpha_n^{-d} \qquad \mbox{for all $n$} \end{equation}
Furthermore, let $\{\tau_n\}_{n \geq 1}$ be stopping times such that for $t \in [0, \tau_n]$ and $n \geq 1$,
\begin{equation} \label{condc3} S_n(t) \leq -n \delta X_n(t) + \beta_n C_2 + C_3 \qquad \mbox{with $\delta>0$} \end{equation}
\begin{equation} \label{condc4} \sup_{\omega\in\Omega, y\in\mathcal{Y}, t\leq\tau_n}\vert f_n(t,y) \vert \leq C_4 \alpha_{n}^{-1},\end{equation}
and also assume
\begin{equation}\label{condc5}
\int_{\mathcal{Y}} (f_n(t,y))^2 A_n(t,dy) \leq C_5\,.
\end{equation}
Then, for any $\varepsilon >0$, there exist $C_6>0$ and $n_0$ such that
\begin{equation}\label{tesithmCOEI}
\sup_{n \geq n_0} \mathcal{P} \left\{ \sup_{0 \leq t \leq T \wedge \tau_n} X_n (t) > C_6 (n^{1/d} \alpha_n^{-1} \vee \alpha_n n^{-1})\right\} \leq \varepsilon 
\end{equation}

\end{theorem}

\section{Auxiliary computations}

\begin{lemma} \label{lemma:collapsing}
Let $\a \in (0,1)$ and, for $m>0$ let $y_m(t)$ be the solution of
\[
\begin{split}
\dot{y}_m(t) & = A + B y^{\a}_m(t) - m y_m(t) \\
y_m(0) & = C,
\end{split}
\]
where $A,B,C$ are positive constants. Then for every $t>0$
\[
\limsup_{m \to \infty} \,m\,y_m(t) <+\infty.
\]
\end{lemma}
\begin{proof}
By comparison principle with the solution of the linear ODE (i.e. $B=0$) one can see that $y_{m}(t)> 0$. Consider the function $\varphi(y) := A + By^{\a} - \frac{m}{2}y$. Note that $\varphi'(y) <0$ for $y>y_m^* := \left(\frac{2\a B}{m}\right)^{\frac{1}{1-\a}}$.
Set now $Y_m := \frac{2(A+B)}{m}$. Note that for $m$ sufficiently large, $Y_m > y_m^*$  and $\varphi(Y_m) \leq 0$, so that $\varphi(y) < 0$ for all $y \geq Y_m$. Let $\tau_{m}=\inf_{t>0}\{y_{m}(t)\leq {Y}_{m}\}$. For $t< \tau_m$
\[
\dot{y}_m(t) = A + B y^{\a}_m(t) - m y_m(t) = \varphi(y_m(t)) - \frac{m}{2} y_m(t) \leq - \frac{m}{2} y_m(t),
\]
so
\begin{equation} \label{esty1}
y_m(t) \leq Ce^{-\frac{m}{2}t}.
\end{equation}
On the other hand, $y_m(t) \leq Y_m$ for $t \geq \tau_m$,  because $\dot y_m(t)\leq 0$ for $t\geq \tau_m$ . Summing up, for all $t>0$ and $m$ sufficiently large,
\[
y_m(t) \leq \frac{2(A+B)}{m} + Ce^{-\frac{m}{2}t},
\]
from which the conclusion follows.

\end{proof}

\begin{theorem}\label{thm:eq.quad}
The stochastic equation
\begin{equation}\label{eq:quad}
dX_{t}=\sqrt{2a X_{t}}  dW_{t} + (-b X^{2}_{t}+c) dt, \quad X_{0}=x_{0}>0,
\end{equation}
with $a,b,c>0$ admits a weak solution 
for which pathwise uniqueness holds. If $0<c<a$, the point $0$ is attainable and the process is reflected upwards at $0$, so that $X$ spends $0$ time at $0$. If $c\geq a$, the point $0$ is not attainable and the process is always strictly positive.

Moreover, the bi-dimensional equation
\begin{equation}\label{eq:quad.bi}
\begin{split}
dX_{t}&=\sqrt{2a X_{t}}  dB_{t} + (-b X^{2}_{t}+c) dt, \quad X_{0}=x_{0}>0, \\
d\xi_{t}&= w dt +  v \sqrt{X_{t}} dW_{t}, \quad \xi_{0}=\bar\xi_{0}\in \R
\end{split}
\end{equation}
with $a,b,c,v>0$, $w\in \R$, $\langle B,W\rangle=\rho \in [-1,1]$ also admits a weak solution, for which pathwise uniqueness holds
\end{theorem}
\begin{proof}
Let us write $\sigma^{2}(x)=2a x$ and $\mu(x)=-b x^{2}+c$.
Following \cite{IW.book}[p.446, Chapter 3. Some results on one-dimensional diffusion processes], we prove that there exist a (weak) unique (pathwise) solution to \eqref{eq:quad}. The diffusion and drift coefficients are of class $\mathcal{C}^1$ on $(0, +\infty)$, the diffusion coefficient is positive on $(0, +\infty)$. So Equation  \eqref{eq:quad} is uniquely defined until the stopping time $T_{e} = T_{0}\wedge T_{\infty}$ where $ T_{0}$ is the hitting time of $0$ and  $T_{\infty}$ the explosion time. Let us now introduce, relative to \eqref{eq:quad},
 \emph{speed measure} $M$
with density $m$ and \emph{scale function} $S$ given by
\begin{equation}
    \label{eq:speend-n-scale}
    m(x)=\frac{2}{\sigma^2(x) s(x)}
\text{ and }
S(x)=\int_1^x s(y) d y,
\end{equation}
with $s(y)= \exp( -\int_1^x 2\mu(y)/\sigma^2(y)d y ) $. We have that $S(+\infty)=+\infty$ for any parameters choice, so the process does not explode, and $S(0)=-\infty$ if $c\geq a$, so in this case  
\cite{IW.book}[Theorem 3.1] tells us that $P(T_{e}=+\infty)=1$, so there exist a (weak) unique (pathwise) solution to \eqref{eq:quad} and with probability one the process remains strictly positive for all times $t$ (the boundary is not attainable).
If $0<c<a$, we have $S(0)\in \R$ and the process hits $0$, i.e. $P(T_{0}<+\infty)=1$. We have $dX_{t}=c dt$, with $c>0$, when $X_{t}=0$, so the process is reflected at $0$ and the explosion time is still infinity, so existence still holds. Pathwise uniqueness can be shown adapting \cite{IW.book}[Chapter IV - Theorem 3.2], using the fact that $-bX^{2}<0$. 
  To obtain a solution to the bi-dimensional equation, the second component must be given by
  \[
\xi_{t}= \xi_{0}  +w t +  v \int_{0}^{t}\sqrt{X_{s}} dW_{s}
  \]
  To have weak existence and pathwise uniqueness, we check that
  \[
E  \int_{0}^{t}(\sqrt{X_{s}})^{2} ds = \int_{0}^{t} E[ X_{s}]ds <+\infty
  \]
  by comparison with a suitable CIR process (or computing it using speed measure and scale function). Therefore the statement holds true.
  \end{proof}

\bibliographystyle{plain}
\bibliography{biblio}

\end{document}